\documentclass[journal]{IEEEtran}
\usepackage{amsthm}
\usepackage{amssymb}
\usepackage{graphicx}
\usepackage{newtxmath}
\usepackage[cal=cm]{mathalpha} 
\usepackage{bm}
\newtheorem{theorem}{Theorem}{}
{}
{}
{}
{}
\newtheorem*{remark}{Remark}
\usepackage{algorithm}
\usepackage{algpseudocode}
\usepackage{graphics}
\usepackage{multirow}
\usepackage{array}
\usepackage{cite}
\usepackage{stfloats}
\usepackage{booktabs}
\usepackage{color}
\usepackage{tabu} 
\usepackage{stfloats}
\usepackage{stmaryrd} %CP Decomposition, \llbracket \rrbracket 
\usepackage[colorlinks,linkcolor=blue,anchorcolor=blue,citecolor=blue,bookmarks=true]{hyperref}
\usepackage{breqn}
\usepackage{placeins}
\usepackage{subcaption}
\usepackage{xcolor}

\makeatletter
\let\sum\relax
\let\prod\relax
\DeclareSymbolFont{largesymbols}{OMX}{cmex}{m}{n}
\DeclareMathSymbol{\sum}{\mathop}{largesymbols}{"50}
\DeclareMathSymbol{\prod}{\mathop}{largesymbols}{"51}
\usepackage[capitalize]{cleveref} % 导言区
\crefname{figure}{Fig.}{Figs.}   % 小写引用：\cref → "Fig. 1", "Figs. 1 and 2"
\Crefname{figure}{Fig.}{Figs.}   % 大写引用：\Cref → "Fig. 1"（句首时）

\usepackage{caption}
\begin{document}
	\title{Self-Calibration DOA Estimation for Moving Dilated Arrays with Gain and Phase Errors}
	% author names and affiliations
	% use a multiple column layout for up to three different
	% affiliations

	\author{\IEEEauthorblockN{
			Kaiyuan Fang,
			Chengzhi Ye, \IEEEmembership{Student Member,~IEEE}, 
			Ruoyu Zhang, \IEEEmembership{Senior Member,~IEEE}, 
			Shichao Wei, 
			\\ Wen Wu, 
			\IEEEmembership{Senior Member,~IEEE}
		}
		\thanks{
		Kaiyuan Fang, Chengzhi Ye, Ruoyu Zhang, Shichao Wei, and Wen Wu are with the Key Laboratory of Near-Range RF Sensing ICs and Microsystems, Ministry of Education, School of Electronic and Optical Engineering, Nanjing University of Science and Technology, Nanjing 210094, China.
		Chengzhi Ye and Shichao Wei are also with Qian Xuesen College, Nanjing University of Science and Technology, Nanjing 210094, China (e-mail: kyfang@njust.edu.cn; ryzhang19@njust.edu.cn; qxsycz8166@njust.edu.cn; weishichao@njust.edu.cn; wuwen@njust.edu.cn).
		 \textit{(Corresponding author: Ruoyu Zhang.)}}
		%\thanks
		%\thanks
		\vspace{-1.55em}
	}

	% use for special paper notices
	%\IEEEspecialpapernotice{(Invited Paper)}

	% make the title area
	\maketitle

	% As a general rule, do not put math, special symbols or citations
	% in the abstract
	%\begin{abstract}
	%f
	%\end{abstract}

	%% no keywords
	%\begin{IEEEkeywords}
	%Sparse vector coding, 
	%\end{IEEEkeywords}

	% For peer review papers, you can put extra information on the cover
	% page as needed:
	% \ifCLASSOPTIONpeerreview
	% \begin{center} \bfseries EDICS Category: 3-BBND \end{center}
	% \fi
	%
	% For peerreview papers, this IEEEtran command inserts a page break and
	% creates the second title.
	\IEEEpeerreviewmaketitle

	\begin{abstract}
	Moving dilated arrays (MDAs) can increase the uniform degrees of freedom (uDOFs) while alleviating mutual coupling, but their performance is susceptible to unknown gain and phase errors (GPEs). To deal with this problem, we investigate self-calibration direction-of-arrival (DOA) estimation for MDAs with unknown GPEs. First, multi-position observations are phase aligned to construct an equivalent virtual array with an enlarged aperture. Next, the motion-induced relationships among the GPEs are exploited to establish explicit calibration constraints. By leveraging the orthogonality between the steering vectors and the noise subspace, the joint estimation of the DOAs and GPEs is formulated as a constrained optimization problem. To solve this problem, a two-stage alternating optimization (AO) algorithm is developed. In the first stage, the GPEs are estimated in closed form with the DOAs fixed. In the second stage, the estimated GPEs are used to calibrate the covariance matrix, after which the DOAs are refined using the multiple signal classification (MUSIC) algorithm. To obtain reliable initial estimates, a rotational invariance-like (RIL) initialization method is further developed by exploiting two translated virtual subarrays with identical GPEs. Furthermore, the Cramér–Rao bound (CRB) is derived to characterize the theoretical performance limit of the considered system. Simulation results demonstrate that the proposed method converges reliably and achieves accurate DOA estimation under different array apertures, numbers of snapshots, GPE levels, and numbers of sources.
\end{abstract}

\begin{IEEEkeywords}
	DOA estimation, moving dilated array, array self-calibration, gain and phase errors, alternating optimization.
\end{IEEEkeywords}

	\section{Introduction}

Direction-of-arrival (DOA) estimation is an important technique in array signal processing and has been widely applied in industrial Internet of Things (IoT), intelligent transportation, autonomous driving and smart manufacturing \cite{526899,10522654,10950424,ZhangRYDOA,JiangMAAssisted,YeRotable,ZhangRYISAC}. Existing DOA estimation methods typically rely on large-scale antenna arrays or sparse arrays to form beams with high angular resolution \cite{HuMIMODeepLearn,YTMIMOGPE,ZhangRYMIMO}. However, these methods generally assume fixed uniform linear arrays with constrained apertures. Therefore, achieving accurate and high-resolution DOA estimation with limited array resources remains an important research problem in array signal processing \cite{10529957}.

Over the past decades, a variety of DOA estimation methods have been developed, including conventional beamforming, maximum-likelihood estimation, and subspace-based methods represented by the multiple signal classification (MUSIC) and the estimation of signal parameters via rotational invariance techniques (ESPRIT) \cite{10144718}. Conventional beamforming methods are simple to implement but generally provide limited angular resolution\cite{665}, whereas maximum-likelihood estimators can achieve high estimation accuracy at the cost of computationally intensive nonconvex optimization\cite{8809740}. In contrast, MUSIC and ESPRIT provide high-resolution estimation with moderate computational complexity and have therefore been widely adopted\cite{8212844}. Existing DOA estimation methods generally employ large scale or sparse antenna arrays to synthesize beams with high angular resolution \cite{5456168}. However, a potential limitation of these methods is their reliance on fixed arrays, whose geometries and apertures cannot be flexibly adjusted after deployment. This lack of spatial flexibility restricts the achievable angular resolution and limits their adaptability to varying sensing requirements. Therefore, overcoming the inherent constraints of fixed arrays is essential for further improving DOA estimation performance.

Recently, moving dilated arrays (MDAs) have attracted research interest. The basic idea of MDAs is to coherently combine measurements collected at different time instants or spatial positions into an equivalent array with an enlarged virtual aperture, thereby improving the spatial resolution and increasing the number of resolvable sources. Recent studies related to MDAs  have mainly sought to increase the uniform degrees of freedom (uDOFs) through array dilation and motion-assisted synthesis, while exploiting the enlarged physical inter-element spacings to alleviate mutual coupling \cite{LishuangDOF,QinGDDOF,ChenNA,YeGlobalCom,LishuangDAfamily,Lishuangq-dilated,9773022}.  
Several array configurations have been developed to increase the uDOFs while alleviating mutual coupling \cite{QinGDDOF,ChenNA,ZhouMUDNA}. In \cite{QinGDDOF}, a dilated nested array was combined with short-distance motion to construct a hole-free difference coarray. By employing multi-subarray dilation and element relocation, Chen \emph{et al.} \cite{ChenNA} improved the balance between uDOFs and mutual coupling. More recently, the multi-subarray dilated nested array developed by L. Zhou and G. Zhou \cite{ZhouMUDNA} redistributes redundant sensors in densely spaced subarrays to larger inter-element spacings, thereby improving DOA estimation performance, particularly under strong mutual coupling.
Beyond nested-array configurations, the extension of array dilation from nested arrays to arbitrary linear arrays was investigated in \cite{LishuangDAfamily}. The resulting arrays increase the uDOFs while reducing mutual coupling, although their uDOF gain becomes limited for dilation factors larger than three. This limitation was subsequently addressed in \cite{Lishuangq-dilated} by introducing a multi-shift synthesis strategy, which enables the achievable uDOFs to scale with  the dilation factor $q$ for general dilation factors.

The aforementioned studies mainly improve the conventional difference coarray through array-geometry design and motion-pattern optimization. From a different coarray-modeling perspective, Guo \emph{et al.} \cite{GuoHDHighOrder} extended the dilated-array framework to high-order virtual coarray models and introduced a sparse shift sequence to reduce the redundancy of moving-array synthesis and further increase the uDOFs.
 The MDA methods reviewed above generally assume accurately calibrated sensor gains and phases. In practice, unknown gain and phase errors (GPEs) caused by hardware imperfections and calibration errors distort the array steering vectors, leading to virtual array manifold mismatch and degraded DOA estimation accuracy. Such performance degradation can be particularly pronounced in MDA systems. 

Various methods have been developed to address array GPEs, including joint parameter estimation, exploitation of partial calibration information and higher-order-statistics- and optimization-based techniques \cite{WFAO,LiuAFWF,LiaoB1,WenGPEMIMO,YeRobust,YeSelfCalibtate,YetCalibrateNA,YeMAGPE,Convex}. These methods mitigate the effect of array manifold mismatch on DOA estimation by estimating the error parameters, constructing error-insensitive quantities, or exploiting partial calibration information.
 For partly calibrated uniform linear arrays (ULAs), Liao and Chan\cite{LiaoB1} extended ESPRIT and obtained closed-form estimates of the DOAs and unknown GPEs. Normalized rotational invariance was exploited by Wen \emph{et al.}    \cite{WenGPEMIMO} to jointly estimate the directions of departure, DOAs, and GPEs using only one well-calibrated element in each of the transmit and receive arrays. Using the rotational invariance of an optimized moving array, Ye \emph{et al.} \cite{YeRobust} estimated partial GPEs and subsequently performed DOA estimation using the calibrated covariance matrix.
 Higher-order statistics and optimization techniques have also been employed to handle GPEs. Selected entries of the fourth-order cumulant matrix were used in \cite{YetCalibrateNA} to recursively estimate the GPEs. For noncoherent subarrays, Tirer and Bialer \cite{Convex} formulated the estimation and correction of inter-subarray phase offsets as a convex optimization problem. A distributionally robust optimization framework was adopted in \cite{LiGPEConvex} to model unknown GPEs, under which the resulting robust DOA estimation problem was reformulated as a tractable convex program. 
More recently, a data-driven method was developed by Labbaf \emph{et al.} \cite{DNN} to learn the mapping from covariance matrices of a uniform circular array under GPEs, mutual coupling, and other array imperfections to the DOAs, achieving accurate and robust DOA estimation.
However, most existing methods focus on fixed arrays, whose achievable aperture and spatial DOFs are constrained by their physical configurations and are difficult to further extend through array motion. Hence, DOA estimation methods specifically designed for MDA systems with unknown GPEs remain necessary.

In this paper, we investigate the self-calibration DOA estimation problem for MDA systems with unknown GPEs. The main contributions of this paper are summarized as follows:\
\begin{itemize}

	\item We propose an MDA signal model in the presence of unknown GPEs. Measurements collected at different positions are aligned through phase compensation to synthesize an equivalent virtual array with an enlarged aperture. By further exploiting the pairwise GPE equality relationships induced by array motion, we formulate the associated constraints and the joint DOA and GPE estimation problem.

	\item We develop a two-stage self-calibration alternating optimization (AO) algorithm. In the first stage, with the DOA estimates fixed, GPE estimation is formulated as a constrained least-squares problem, for which a closed-form solution is obtained using the Lagrange multiplier method. In the second stage, the covariance matrix is calibrated using the estimated GPEs, followed by DOA updating via MUSIC. These two stages are performed alternately to progressively improve both GPE calibration and DOA estimation. Moreover, a rotational invariance-like (RIL) initialization is further designed to provide reliable initial DOA estimates.

	\item We analyze the theoretical performance and computational complexity of the proposed method. Specifically, the Fisher information matrix and the corresponding Cramér--Rao bound (CRB) for the joint estimation of DOAs and GPEs are derived, providing a theoretical benchmark for DOA estimation accuracy in the presence of unknown GPEs. In addition, the computational complexities of the main algorithmic steps and the overall iterative procedure are analyzed to evaluate the computational cost and practical feasibility of the proposed method.

	\item Extensive simulations verify the effectiveness of the proposed method. The results demonstrate that the proposed algorithm converges reliably in the presence of unknown GPEs and achieves DOA estimation accuracy close to the Cramér--Rao bound. They also reveal the relationship between GPE calibration accuracy and DOA estimation performance and confirm the robustness of the proposed method under different array apertures, numbers of snapshots, error levels, and numbers of sources.

\end{itemize}

\noindent\textit{Notations:}
Scalars, vectors, matrices, and sets are denoted by italic, boldface lowercase, boldface uppercase, and calligraphic uppercase letters, respectively. 
$\mathbb{R}$ and $\mathbb{C}$ denote the real and complex fields, respectively, and $j=\sqrt{-1}$ is the imaginary unit. 
$(\cdot)^T$, $(\cdot)^*$, $(\cdot)^H$, and $(\cdot)^{-1}$ represent the transpose, conjugate, conjugate transpose, and inverse, respectively. 
$[\bm{A}]_{p,q}$ denotes the $(p,q)$-th entry of $\bm{A}$, while $\bm{A}(\mathcal{P},:)$ denotes the submatrix formed by the rows indexed by $\mathcal{P}$.  
 $\bm{0}_{M\times N}$ and $\bm{1}_{M\times N}$ denote the
$M\times N$ zero matrix and all-ones matrix, respectively.
$\bm{I}_N$ denotes the $N\times N$ identity matrix. The vectors
$\bm{i}_k\in\mathbb{R}^{K\times 1}$ and
$\bm{q}_m\in\mathbb{R}^{M\times 1}$ denote the $k$-th and $m$-th
standard basis vectors, respectively.
$\|\cdot\|$, $\mathbb{E}\{\cdot\}$, $\operatorname{Cov}(\cdot)$, $\operatorname{diag}(\cdot)$, and $\operatorname{tr}(\cdot)$ denote the Euclidean norm, expectation, covariance, diagonalization operator, and trace, respectively. 
$\mathfrak{R}\{\cdot\}$ extracts the real part, while $\circledast$ denotes the Hadamard product.

\section{System Model }

In this section, we propose a DOA estimation model for an MDA system with GPEs. As illustrated in Fig.~\ref{model}, a ULA mounted on a moving platform collects observations at multiple spatial positions. These observations are carrier phase synchronized and spatially rearranged to synthesize an equivalent virtual array with an enlarged aperture.
We assume that the source complex envelopes, DOAs, and
propagation environment remain approximately unchanged over
one virtual array synthesis interval. 
 The first antenna is assumed to be calibrated and its GPE coefficient remains unity throughout the array movement.
 Based on the resulting virtual array model, we then derive its received signal and covariance models. Finally,  we formulate DOA estimation in the presence of unknown GPEs as a joint optimization problem over the DOAs and GPEs.

\begin{figure}[h]
	\centering
	\includegraphics[width=0.9\linewidth]{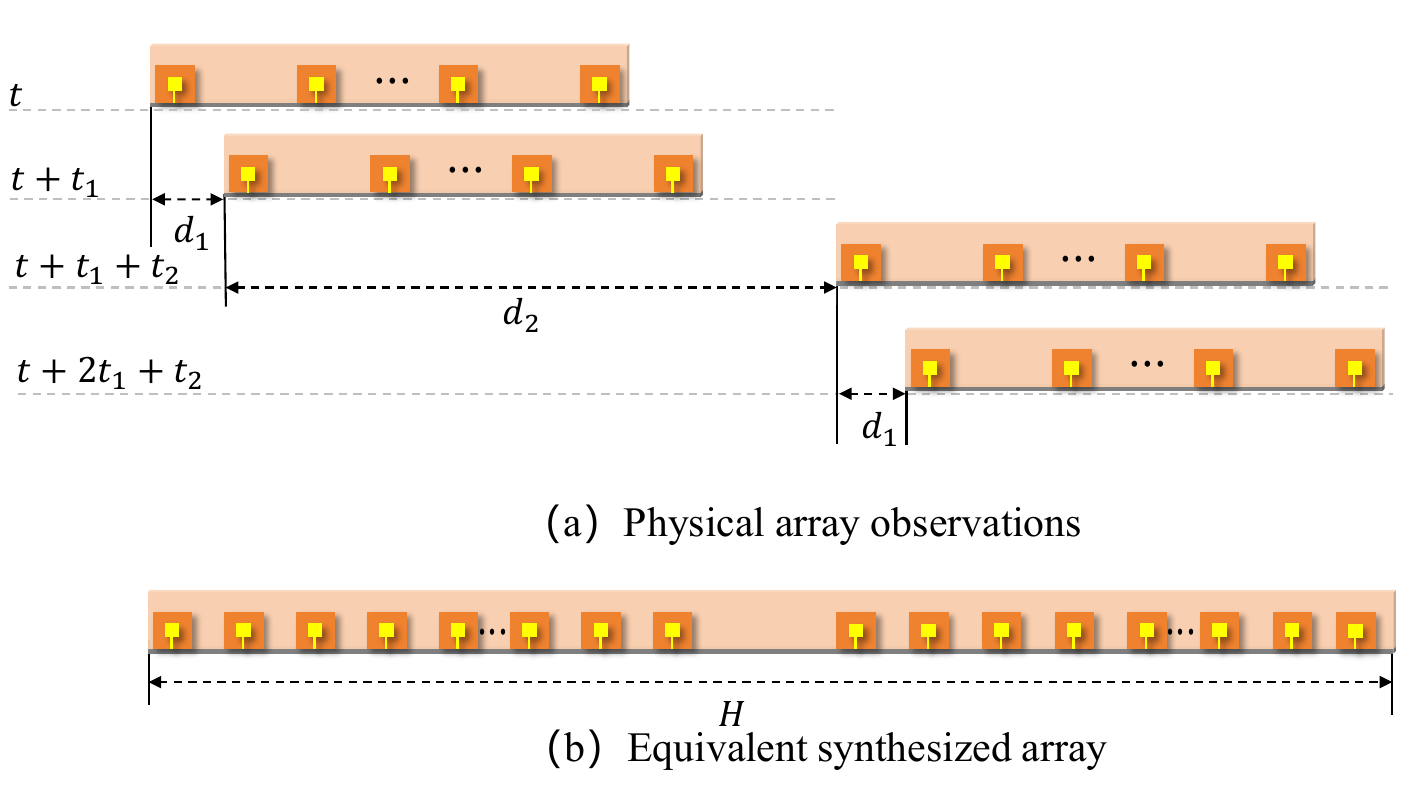}
	\caption{Equivalent synthesized array model based on multi-position observations.}
	\label{model}
\end{figure}
At time $t$, we consider a ULA mounted on a platform that moves along the array axis at a constant velocity $v$. The array consists of $M_0$ antenna elements with an inter-element spacing of $d=\lambda$ and receives signals from $K$ mutually uncorrelated narrowband far-field sources. At the reference time $t$, let $s_k(t)$ denote the signal emitted by the $k$-th source. The remaining receive channels are subject to unknown GPEs. 
The received signal vector is
\begin{align}
	\begin{aligned}
		\bm{y}(t) &= \sum_{k=1}^{K}  \bm{\varGamma}_1 s_k(t) e^{j2\pi vt\frac{\cos(\theta_k)}{\lambda}} \bm{a}_p(\theta_k) + \bm{n}(t) \\
		&= \bm{\varGamma}_1 \bm{A}_{p0} \bm{s}(t) + \bm{n}(t),
	\end{aligned}
\end{align}
where  $\lambda$ denotes the carrier wavelength, $\theta_k$ is the DOA of the $k$-th source, $\bm{a}_p(\theta_k)\in\mathbb{C}^{M_0\times1}$ is the steering
vector associated with $\theta_k$, given by
\begin{align}
	\bm{a}_p(\theta_k)=\big[1,e^{j2\pi\frac{d\cos(\theta_k)}{\lambda}},e^{j2\pi\frac{2d\cos(\theta_k)}{\lambda}},\ldots,e^{j2\pi\frac{(M_0-1)d\cos(\theta_k)}{\lambda}}\big]^T, \nonumber
\end{align}
$\bm{n}(t)\in\mathbb{C}^{M_0\times1}$ is the additive white Gaussian noise (AWGN) vector with $\bm{n}(t)\sim\mathcal{CN}\left(\bm{0}_{M_0\times1},\sigma_n^2\bm{I}_{M_0}\right)$ and $\sigma_n^2$ denotes the noise variance.
 Moreover, $\bm{\varGamma}_1=\operatorname{diag}(\rho_1e^{j\varphi_1},\rho_2e^{j\varphi_2},
\ldots,\rho_{M_0}e^{j\varphi_{M_0}})\in\mathbb{C}^{M_0\times M_0}$ denotes the GPE matrix, $\rho_m$ and $\varphi_m$ denote the gain coefficient and phase error of the $m$-th physical antenna element, respectively. $\bm{A}_{p0}=[\bm{a}_p(\theta_1),\bm{a}_p(\theta_2),\ldots,\bm{a}_p(\theta_K)]\in\mathbb{C}^{M_0\times K}$ is the steering matrix of the physical array, and
$\bm{s}(t)=[s_1(t)e^{j2\pi vt\frac{\cos(\theta_1)}{\lambda}},s_2(t)e^{j2\pi vt\frac{\cos(\theta_2)}{\lambda}},\ldots,
s_K(t)e^{j2\pi vt\frac{\cos(\theta_K)}{\lambda}}]^T\in\mathbb{C}^{K\times1}$
is the source signal vector.

At time $t+t_1$, the platform moves to the right by
$d_1=vt_1=\lambda/2$. Since this displacement is small, the GPEs are assumed to remain unchanged. The received signal vector becomes
\begin{align}
	\bm{y}(t+t_1) &= \sum_{k=1}^{K} \bm{\varGamma}_1 s_k(t+t_1)e^{j2\pi vt\frac{\cos(\theta_k)}{\lambda}} \nonumber\\
	&\quad\times e^{j2\pi vt_1\frac{\cos(\theta_k)}{\lambda}}\bm{a}_p(\theta_k)+\bm{n}(t+t_1) \nonumber\\
	&=\bm{\varGamma}_1\bm{A}_{p1}\bm{s}(t+t_1)+\bm{n}(t+t_1),
	\label{y_t+t1}
\end{align}
where
$\bm{A}_{p1}=[\bm{a}_p(\theta_1)e^{j2\pi\frac{d_1\cos(\theta_1)}{\lambda}},
\bm{a}_p(\theta_2)e^{j2\pi\frac{d_1\cos(\theta_2)}{\lambda}},
\ldots,\\
\bm{a}_p(\theta_K)e^{j2\pi\frac{d_1\cos(\theta_K)}{\lambda}}]\in\mathbb{C}^{M_0\times K}$,
and
$\bm{s}(t+t_1)=[s_1(t+t_1)e^{j2\pi vt\frac{\cos(\theta_1)}{\lambda}},
s_2(t+t_1)e^{j2\pi vt\frac{\cos(\theta_2)}{\lambda}},
\ldots,
s_K(t+t_1)e^{j2\pi vt\frac{\cos(\theta_K)}{\lambda}}]^T\in\mathbb{C}^{K\times1}$.
The noise vectors $\bm{n}(t)$ and $\bm{n}(t+t_1)$ are assumed to be uncorrelated.

Assume that the complex envelopes of the narrowband sources remain approximately unchanged during the entire aperture-synthesis interval and that the carrier frequency and movement duration are accurately known. The temporal variation of the source signal vector is therefore dominated by the deterministic carrier-phase progression, i.e.,
\begin{align}	
	\bm{s}(t+\tau) \approx e^{j 2\pi f \tau} \bm{s}(t), \tau \in \{t_1, t_1+t_2, 2t_1+t_2\},
\end{align}
where $f$ denotes the carrier frequency, and $\tau$ is the elapsed time from the reference instant $t$. 
 Accordingly, multiplying $\bm{y}(t+t_1)$ by $e^{-j2\pi f t_1}$ yields the phase-synchronized received signal vector\cite{QinGDDOF} 
\begin{align}
	\tilde{\bm{y}}_1(t+t_1)
	&=\bm{y}(t+t_1)e^{-j2\pi ft_1} \nonumber\\
	&=\bm{\varGamma}_1\bm{A}_{p1}\bm{s}(t)
	+\tilde{\bm{n}}(t+t_1),
\end{align}
where
$
\tilde{\bm{n}}(t+t_1)
=
e^{-j2\pi ft_1}
\bm{n}(t+t_1).
$

Following the same carrier-phase compensation procedure, the phase-synchronized observations collected after the subsequent two movements are given as follows.

At time $t+t_1+t_2$, the platform moves further to the right by
$d_2=vt_2$. Since this displacement is relatively large, the GPE
matrix changes to $\bm{\varGamma}_2$. After the carrier-phase compensation procedure,
the phase-synchronized
received signal vector is
\begin{align}
	\tilde{\bm{y}}_2(t+t_1+t_2)
	=\bm{\varGamma}_2
	\bm{A}_{p2}
	\bm{s}(t)
	+\tilde{\bm{n}}(t+t_1+t_2),
	\label{y2_phase_aligned}
\end{align}
where $
\bm{A}_{p2}
=
\big[
\bm{a}_p(\theta_1)
e^{j2\pi\frac{(d_1+d_2)\cos(\theta_1)}{\lambda}},
\bm{a}_p(\theta_2)
e^{j2\pi\frac{(d_1+d_2)\cos(\theta_2)}{\lambda}},
\\\ldots,
\bm{a}_p(\theta_K)
e^{j2\pi\frac{(d_1+d_2)\cos(\theta_K)}{\lambda}}
\big]
\in\mathbb{C}^{M_0\times K}
$,
$
\tilde{\bm{n}}(t+t_1+t_2)
=
e^{-j2\pi f(t_1+t_2)}
\bm{n}(t+t_1+t_2).
$

At time $t+2t_1+t_2$, the platform moves by another short
distance $d_1=vt_1$. Therefore, the GPE matrix remains
unchanged. 
After the carrier-phase compensation procedure, the phase-synchronized
received signal vector is
\begin{align}
	\tilde{\bm{y}}_3(t+2t_1+t_2)
	&=
	e^{-j2\pi f(2t_1+t_2)}
	\bm{y}(t+2t_1+t_2)
	\nonumber\\
	&=
	\bm{\varGamma}_2
	\bm{A}_{p3}
	\bm{s}(t)
	+\tilde{\bm{n}}(t+2t_1+t_2),
	\label{y3_phase_aligned}
\end{align}
where $
\bm{A}_{p3}
=
\big[
\bm{a}_p(\theta_1)
e^{j2\pi\frac{(2d_1+d_2)\cos(\theta_1)}{\lambda}},
\bm{a}_p(\theta_2)
e^{j2\pi\frac{(2d_1+d_2)\cos(\theta_2)}{\lambda}},
\\\ldots,
\bm{a}_p(\theta_K)
e^{j2\pi\frac{(2d_1+d_2)\cos(\theta_K)}{\lambda}}
\big]
\in\mathbb{C}^{M_0\times K}
$,
$
\tilde{\bm{n}}(t+2t_1+t_2)
=
e^{-j2\pi f(2t_1+t_2)}
\bm{n}(t+2t_1+t_2).
$

After the carrier-phase synchronization, the  observations share the same source-signal vector $\bm{s}(t)$, while retaining the displacement-induced spatial phase shifts in their respective position-dependent steering matrices. As illustrated in Fig.~\ref{model}, these observations are subsequently stacked and rearranged according to the spatial locations of the virtual sensors, thereby synthesizing an equivalent virtual array with an enlarged aperture.

Consequently, the equivalent virtual array consists of
$M=4M_0$ virtual elements located at

\begin{align}
	\mathcal{P}
	=&\left\{(n-1)d_1,\ 1\le n\le 2M_0\right\}\nonumber\\
	&\cup\left\{
	H-(M-m)d_1,\ 2M_0+1\le m\le M
	\right\},
\end{align}
where
\begin{align}
	H
	=v(2t_1+t_2)+(M_0-1)d
	=2d_1+d_2+(M_0-1)d \nonumber
\end{align}
denotes the total array length.

Owing to the array movement, the odd- and even-indexed entries
of the GPE vector $\bm{\gamma}=[\gamma_1,\ldots,\gamma_M]^T$ satisfy
\begin{align}
	\gamma_m=\gamma_{m+1},\quad
	m=1,3,\ldots,\frac{M}{2}-1,
	\frac{M}{2}+1,\ldots,M-1.
	\label{even=odd}
\end{align}
Since the first of the $M_0$ physical antennas is assumed to remain calibrated throughout the array movement, the corresponding virtual elements satisfy
\begin{align}
	\gamma_1=\gamma_2
	=\gamma_{M/2+1}
	=\gamma_{M/2+2}=1.
	\label{odd1111}
\end{align}

To rearrange the stacked observations according to the spatial order of the equivalent virtual sensors, we introduce the permutation matrix $\bm{J}$. Let $\bm{q}_\jmath$ denote the $\jmath$ -th column of the identity matrix $\bm{I}_M$. The permutation matrix is then constructed as
\begin{align}
 \bm{J}
=
 \begin{bmatrix}
	\bm{q}_1,
	\bm{q}_{2M_0+1} ,
	\bm{q}_2 ,
	\bm{q}_{2M_0+2} ,
	\cdots ,
	\bm{q}_{M/2} ,
	\bm{q}_{M}
\end{bmatrix} ^{T}.	
\end{align}

After stacking and spatially rearranging the  phase-aligned
observations, the output of the equivalent virtual array at time
$t$ is
\begin{align}
	\bm{z}(t)
	 =\bm{J}
	\begin{bmatrix}
		\bm{y}(t)\\
		\tilde{\bm{y}}_2(t+t_1+t_2)\\
		\tilde{\bm{y}}_1(t+t_1)\\
		\tilde{\bm{y}}_3(t+2t_1+t_2)
	\end{bmatrix}\nonumber
	 =\bm{\varGamma}\bm{A}\bm{s}(t)+\bm{\varepsilon}(t),
\end{align}
where  $\bm{\varepsilon}(t)\in\mathbb{C}^{M\times1}$ is the
AWGN vector with
$\operatorname{Cov}[\bm{\varepsilon}(t)]
=\sigma_n^2\bm{I}_M$, the equivalent GPE matrix $\bm{\varGamma}$ is
\begin{align}
	\bm{\varGamma}
	=\bm{J}\operatorname{blkdiag}\left\{
	\bm{\varGamma}_1,
	\bm{\varGamma}_2,
	\bm{\varGamma}_1,
	\bm{\varGamma}_2
	\right\}\bm{J}^{T}
	=\operatorname{diag}\{\bm{\gamma}\}
	\in\mathbb{C}^{M\times M},
\end{align}
and the equivalent steering matrix is explicitly given by
\begin{align}
	\bm{A}
	=\bm{J}
	[\bm{A}_{p0}^T,
	\bm{A}_{p2}^T,
	\bm{A}_{p1}^T,
	\bm{A}_{p3}^T]^T
	\in\mathbb{C}^{M\times K}.
	\label{equivalent_manifold}
\end{align}

At each position, $T$ snapshots are collected.
After carrier-phase synchronization, the observations with the
same snapshot index are stacked to construct $T$ equivalent
virtual-array snapshots. Accordingly, the received-signal matrix
is expressed as
\begin{align}
	\bm{Z}
	=\bm{\varGamma}\bm{A}\bm{S}+\bm{N}
	=\bm{G}\bm{D}_{\varphi}\bm{A}\bm{S}+\bm{N},
	\label{Z}
\end{align}
where
$\bm{\varGamma}=\operatorname{diag}\{\rho_1e^{j\varphi_1},\rho_2e^{j\varphi_2},
\ldots,\rho_Me^{j\varphi_M}\}
\in\mathbb{C}^{M\times M}$,
$\bm{G}=\operatorname{diag}\{\rho_1,\rho_2,\ldots,\rho_M\}
\in\mathbb{C}^{M\times M}$,
$\bm{D}_{\varphi}
=\operatorname{diag}\{e^{j\varphi_1},e^{j\varphi_2},\ldots,e^{j\varphi_M}\}
\in\mathbb{C}^{M\times M}$,
and
$\bm{S}
=[\bm{s}(1),\bm{s}(2),\ldots,\bm{s}(T)]
\in\mathbb{C}^{K\times T}$. The noise matrix is given by
	$\bm{N}=[\bm{\varepsilon}(1),\bm{\varepsilon}(2),
	\ldots,\bm{\varepsilon}(T)]
	\in\mathbb{C}^{M\times T}$.

To facilitate the subsequent subspace-based DOA estimation,
we next characterize the second-order spatial statistics
of the equivalent virtual-array observations. The covariance matrix of $\bm{Z}$
is given by
\begin{align}
	\bm{R}
	=\mathbb{E}\left\{\bm{Z}\bm{Z}^H\right\}
	=\bm{\varGamma}\bm{A}
	 \bm{R}_S
	\bm{A}^H\bm{\varGamma}^H
	+\sigma_n^2\bm{I}_M,
	\label{Finalmodel}
\end{align}
where
 $	\bm{R}_S=\mathbb{E}\{\bm{S}\bm{S}^H\}
=\operatorname{diag}
(\sigma_1^2,\sigma_2^2,\ldots,\sigma_K^2)$
is the source covariance matrix, and $\sigma_k^2$ is the
average power of the $k$-th source.
 To obtain the signal and noise subspaces required for subsequent DOA estimation,
we perform the eigendecomposition of $\bm{R}$ as
\begin{align}
	\bm{R}
	=\sum_{m=1}^{M}\upsilon_m\bm{e}_m\bm{e}_m^H
	=\bm{E}_S\bm{\varUpsilon}_S\bm{E}_S^H
	+\bm{E}_N\bm{\varUpsilon}_N\bm{E}_N^H,
	\label{R}
\end{align}
where $\upsilon_m$ is the $m$-th largest eigenvalue and
$\bm{e}_m$ is its associated eigenvector. The signal- and
noise-subspace matrices are
$\bm{E}_S=[\bm{e}_1,\ldots,\bm{e}_K]
\in\mathbb{C}^{M\times K}$ and
$\bm{E}_N=[\bm{e}_{K+1},\ldots,\bm{e}_M]
\in\mathbb{C}^{M\times(M-K)}$, respectively. 
$\bm{\varUpsilon}_S
=\operatorname{diag}\{\upsilon_1,\upsilon_2,\ldots,\upsilon_K\}$
and
$\bm{\varUpsilon}_N
=\operatorname{diag}
\{\upsilon_{K+1},\upsilon_{K+2},\ldots,\upsilon_M\}$.

For $\theta_k\in\{\theta_1,\theta_2,\ldots,\theta_K\}$,
the orthogonality between the signal and noise subspaces yields
\begin{align}
	\bm{E}_N^H\bm{\varGamma}\bm{a}(\theta_k)
	=\bm{0} _{(M-K)\times 1},
\end{align}
where the steering vector of the equivalent virtual array
associated with the $k$-th source is defined as
\begin{align}
	\bm{a}(\theta_k)
	=
	\left[
	e^{j2\pi\frac{p_1\cos(\theta_k)}{\lambda}},
	e^{j2\pi\frac{p_2\cos(\theta_k)}{\lambda}},
	\ldots,
	e^{j2\pi\frac{p_M\cos(\theta_k)}{\lambda}}
	\right]^T,
	\label{a} \nonumber
\end{align}
 and $p_m\in\mathcal{P}$ denotes the position
of the $m$-th virtual element.

Our objective is to estimate the DOAs in an MDA system in the
presence of GPEs. However, the GPEs distort the original array
manifold, inevitably degrading the DOA estimation accuracy.
To address this issue, we reformulate the DOA estimation problem
in the presence of GPEs as a joint optimization problem of the
DOAs and GPEs, which can be expressed as
\begin{subequations}
	\begin{alignat}{2}
		\text{(P1)}\quad
		\min_{\bm{\theta},\bm{\gamma}}\quad
		&\sum_{k=1}^{K}
		\left\|
		 {\bm{E}} _{N}^{H}\bm{\varGamma}
		\bm{a}\left({\theta}_{k}\right)
		\right\|_F^{2}
		&&
		 \\
		\text{s.t.}\quad
		&\gamma_m= \gamma_{m+1},
		\qquad &m&\in\mathcal{I}_{1},	\label{odd2}
		\\
		&\gamma_m= 1,
		\qquad &m&\in\mathcal{I}_{2},\label{even2}
	\end{alignat}
\end{subequations}
where $\bm{\theta}=[\theta_1,\ldots,\theta_K]^T$ denotes the DOA vector, 
 $
\mathcal{I}_{1}
=\{
3,5,\ldots,M/2-1
\}
\cup
\{
M/2+3,M/2+5,\ldots,M-1
\}
,
$
and 
$
\mathcal{I}_{2}
=\{
1,2,M/2+1,M/2+2
\}
.
$ The constraints in \eqref{odd2} and \eqref{even2} follow
directly from \eqref{even=odd} and \eqref{odd1111},
respectively. 
 It should be noted that the DOAs and GPEs are coupled in the objective function,
making problem (P1) difficult to solve directly. Therefore, in
the next section, we develop a two-stage alternating optimization
algorithm to estimate the DOAs and GPEs alternately.
\section{Proposed Self-Calibration AO Algorithm}

In this section, a self-calibration AO algorithm is developed for DOA estimation, which contains two stages. In Stage 1, given the initial DOA estimates, the original joint optimization problem reduces to a GPE estimation subproblem which is reformulated and solved in closed form using the Lagrange multiplier method. In Stage 2, the estimated GPEs are used to calibrate the received signal covariance matrix, after which the DOAs are updated via the MUSIC algorithm. At each iteration, the updated DOA estimates are used to refine the GPE estimates and covariance matrix calibration, which in turn further improve the DOA estimates in the subsequent iteration. In addition, an RIL initialization method is developed to obtain reliable initial DOAs by exploiting the identical GPEs shared by adjacent odd- and even-indexed virtual elements.

\subsection{Stage 1: GPE Estimation}
\label{stage:gpe}
In this subsection, the DOAs are fixed at their current estimates, such that the DOA estimation problem is reduced to a GPE optimization subproblem. The resulting problem is reformulated as a tractable linearly constrained quadratic program, for which a closed-form solution is derived using the Lagrange multiplier method. 

Given the DOA estimates \(\hat{\bm{\theta}} \) at the \(l\)-th iteration, the joint optimization problem (P1) can be directly reformulated as an estimation problem involving only the GPEs. Then, problem (P1) can be rewritten as
\begin{subequations}
	\begin{alignat}{2}
		\text{(P2)}\quad
		\min_{\bm{\gamma}}\quad
		&\sum_{k=1}^{K}
		\left\|
		 {\bm{E}} _{N}^{H}\bm{\varGamma}
		\bm{a}\left(\hat{\theta}_{k}\right)
		\right\|_F^{2}
		&&
		\label{P2}\\
		\text{s.t.}\quad
		&\gamma_m= \gamma_{m+1},
		\qquad &m&\in\mathcal{I}_{1},\label{P21111}
		\\
		&\gamma_m= 1,
		\qquad &m&\in\mathcal{I}_{2}.
		\label{P2oddeven}
	\end{alignat}
\end{subequations}
Let \(\bm{\varLambda}_k =\operatorname{diag}\!\left(\bm{a}(\hat{\theta}_k )\right)\in\mathbb{C}^{M\times M}\). \eqref{P2} can be rewritten as
\begin{align}
	&\min_{\bm{\gamma}}\quad
	\sum_{k=1}^{K}
	\left\|
	 {\bm{E}} _{N}^{H}\bm{\varGamma}
	\bm{a}\left(\hat{\theta}_{k} \right)
	\right\|_F^{2}\nonumber
	\Leftrightarrow\min_{\bm{\gamma}}\quad
	\sum_{k=1}^{K}
	\left\|
	 {\bm{E}} _{N}^{H}\bm{\varLambda}_k \bm{\gamma}
	\right\|_F^{2}\nonumber\\
	&\Leftrightarrow\min_{\bm{\gamma}}\quad
	\bm{\gamma}^{H}
	\left(
	\sum_{k=1}^{K}
	 \bm{\varLambda}_k  ^{H}
	 {\bm{E}}_{N}{\bm{E}} _{N} ^{H}
	\bm{\varLambda}_k 
	\right)
	\bm{\gamma}.
	\label{min est}
\end{align}
The first equality follows from the identity $\operatorname{diag}\{\bm{x}\}\bm{y}=\operatorname{diag}\{\bm{y}\}\bm{x}$. Since ${\bm{E}}_N^H\bm{\varLambda}_k\bm{\gamma}$ is a column vector, its squared Frobenius norm can be written in the corresponding Hermitian quadratic form, which gives the second equality.

To complete the reformulation of the GPE estimation problem, the constraints in \eqref{P21111} and \eqref{P2oddeven} are further expressed in matrix form as
\begin{align}
	\bm{W}_{0}^{H}\bm{\gamma}=\bm{b},
	\label{Constraint}
\end{align}
where $\bm{W}_0\in\mathbb{C}^{M\times(M/2+2)}$ is given by
\begin{align}
	\bm{W}_0=
	\begin{bmatrix}
		\bm{W}_1^T,
		\bm{W}_2^T
	\end{bmatrix},
\end{align}
where $\bm{W}_1\in\mathbb{C}^{4\times M}$ selects the four calibrated elements of the virtual array, whereas $\bm{W}_2\in\mathbb{C}^{(M/2-2)\times M}$ imposes the identical-GPE constraints on adjacent odd- and even-indexed virtual elements. Their entries are defined as
	\begin{align}
		\bm W_1
		&=
		\begin{bmatrix}
			\bm q_1 &
			\bm q_2 &
			\bm q_{M/2+1} &
			\bm q_{M/2+2}
		\end{bmatrix}^{T},
		\label{W1}\\
		\bm W_2
		&=
		 \begin{bmatrix}
			(\bm q_3-\bm q_4)^T\\
			(\bm q_5-\bm q_6)^T\\
			\vdots\\
			(\bm q_{M/2-1}-\bm q_{M/2})^T\\
			(\bm q_{M/2+3}-\bm q_{M/2+4})^T\\
			\vdots\\
			(\bm q_{M-1}-\bm q_M)^T
		\end{bmatrix},
		\label{W2}
	\end{align}
where $\bm q_m\in\mathbb{R}^{M}$ denotes the $m$-th standard basis vector.

The vector $\bm{b}$ is
\begin{align}
	\bm{b}
	=
	\begin{bmatrix}
		\bm{1}_{4\times1}^{T},
		\bm{0}_{(M/2-2)\times1}^{T}
	\end{bmatrix}^{T}
	\in\mathbb{R}^{(M/2+2)\times1}.
	\label{b}
\end{align}

In practice, the noise subspace matrix $\bm{E}_N$ is estimated
from the sample covariance matrix of the received data, yielding
$\hat{\bm{E}}_N$. We then substitute $\hat{\bm{E}}_N$ for
$\bm{E}_N$ in problem (P2). \begin{subequations}
	\begin{align}
		\text{(P3)}\quad
		\min_{\bm{\gamma}}\quad
		&\bm{\gamma}^{H}
		 \sum_{k=1}^{K}
		 \bm{\varLambda}_k  ^{H}
		\hat{\bm{E}}_{N}\hat{\bm{E}}_{N}^{H}
		\bm{\varLambda}_k 
		 \bm{\gamma}\\
		\text{s.t.}\quad
		&\bm{W}_{0}^{H}\bm{\gamma}=\bm{b}.
	\end{align}
\end{subequations}
Problem (P3) is a quadratic optimization problem with linear equality constraints. The optimal solution to this problem is given in Theorem~\ref{thm:theorem1}.

\begin{theorem}
	\label{thm:theorem1}
	For fixed DOA estimates $\hat{\bm{\theta}} $,
	suppose that the finite-sample matrix
	$\hat{\bm{V}} $ is nonsingular and that
	 $
	\bm{W}_0^H
	\hat{\bm{V}} ^{-1}
	\bm{W}_0
	$
	is nonsingular. Then, the solution to problem~(P3) is\cite{LiaoB1}
	\begin{align}
		\hat{\bm{\gamma}} 
		=
		 \hat{\bm{V}} ^{-1}
		\bm{W}_{0}
		 \left(
		\bm{W}_{0}^{H}
		 \hat{\bm{V}} ^{-1}
		\bm{W}_{0}
		 \right) ^{-1}
		\bm{b},
		\label{gamma}
	\end{align}
	where
	 $\hat{\bm{V}} =\sum_{k=1}^{K}\bm{\varLambda}_k ^H\hat{\bm{E}}_{N}\hat{\bm{E}}_{N}^{H}\bm{\varLambda}_k \in\mathbb{C}^{M\times M}$.
\end{theorem}
\begin{proof}
	See Appendix~\ref{app:proof_theorem1}.
\end{proof}
Theorem~\ref{thm:theorem1} provides an explicit solution to the constrained GPE estimation problem, avoiding the need for an additional iterative solver within Stage 1. This closed-form update makes the GPE estimation computationally convenient for the proposed AO framework.

When the GPE estimation matrix is singular or ill-conditioned,  diagonal loading is applied as $\hat{\bm V}_{\epsilon}=\hat{\bm V}+\epsilon\bm I_M$, where $\epsilon>0$ is a small regularization parameter. Since $\hat{\bm V}$ is positive semidefinite, this modification shifts its eigenvalues by $\epsilon$, yielding a positive-definite matrix.

For fixed DOA estimates, the GPE estimates can be obtained directly from the closed-form expression in Theorem~\ref{thm:theorem1}. Therefore, the GPE   can be directly updated at each AO iteration while satisfying the GPE constraints induced by the array movement. The resulting GPE estimate is subsequently used to calibrate the covariance matrix in Stage 2.

\subsection{Stage 2: DOA Estimation}
\label{stage:doa}
In Stage 2, the estimated GPEs are compensated in the received signal covariance matrix, from which the calibrated noise subspace is obtained. The DOAs are then updated using the MUSIC algorithm and fed back to Stage 1 for the next AO iteration.

We next calibrate the covariance matrix using $\hat{\bm{\gamma}} $ and update the DOAs via MUSIC. From the eigenvalues associated with the noise subspace in \eqref{R}, the noise power is estimated as
\begin{align}
	\hat{\sigma}_n^2
	=\frac{1}{M-K}\sum_{m=K+1}^{M}\upsilon_m.
\end{align}
The noise-subtracted covariance matrix $\tilde{\bm{R}} \in\mathbb{C}^{M\times M}$ is
\begin{align}
	\tilde{\bm{R}} 
	=\bm{R}-\hat{\sigma}_n^2\bm{I}
	\approx
	\bm{\varGamma}\bm{A}
	\mathbb{E}\{\bm{S}\bm{S}^H\}
	\bm{A}^H\bm{\varGamma}^H.
\end{align}
 Using $\hat{\bm{\varGamma}} =\operatorname{diag}(\hat{\bm{\gamma}} )\in\mathbb{C}^{M\times M}$, the calibrated covariance matrix $\hat{\bm{R}} 
\in\mathbb{C}^{M\times M}$ is obtained as
\begin{align}
	\hat{\bm{R}} 
	=
	 \hat{\bm{\varGamma}} ^{-1}
	\tilde{\bm{R}} 
	 \left(
	\hat{\bm{\varGamma}}  ^{H}
	 \right) ^{-1}
	\approx
	\bm{A}\mathbb{E}\{\bm{S}\bm{S}^H\}\bm{A}^H.
	\label{R_l}
\end{align}

After GPE compensation, the calibrated covariance matrix is used to update the DOAs. Specifically, its eigendecomposition is performed to obtain the noise subspace required by the MUSIC algorithm.
  Then we  perform the eigendecomposition of $\hat{\bm{R}} $,
\begin{align}
	\hat{\bm{R}} 
	=\sum_{m=1}^{M}
	\zeta_m \bm{u}_m 
	 \bm{u}_m  ^H,
	\label{R_l_evd}
\end{align}
where $\zeta_m $ is the $m$-th largest eigenvalue and $\bm{u}_m $ is the corresponding eigenvector. The noise-subspace matrix is
\begin{align}
	\tilde{\bm{E}}_N 
	=\left[
	\bm{u}_{K+1} ,
	\bm{u}_{K+2} ,
	\ldots,
	\bm{u}_{M} 
	\right]\in\mathbb{C}^{M\times(M-K)}.
	\label{E_N_l}
\end{align}
The MUSIC signal power spectrum is then computed as
\begin{align}
	 P (\theta)
	=\frac{1}{
		\bm{a}^H(\theta)
		\tilde{\bm{E}}_N 
		\tilde{\bm{E}}_N^H
		\bm{a}(\theta)}.
	\label{P_l}
\end{align}
The $K$ dominant peaks of ${P} (\theta)$ yield the updated DOA estimates $\hat{\bm{\theta}}$. The iterations terminate when the difference between two successive DOA estimates is smaller than $\Delta\theta$.

It should be noted that the DOA estimation in Stage 2 relies on the GPE estimates obtained in Stage 1, whose accuracy depends on the initial DOA estimates in turn. Therefore, sufficiently accurate initialization is essential for reliable GPE estimation and rapid convergence of the proposed AO algorithm. Otherwise, estimation errors may propagate between the two stages, leading to slow convergence, convergence to a local optimum, and large DOA estimation errors. To this end, a reliable DOA RIL initialization method is developed in the following subsection.

\subsection{RIL DOA Initialization}
\label{stage:el}
To obtain reliable initial DOA estimates for the proposed AO algorithm, an RIL initialization method is developed based on the GPE relationships induced by array movement. Two translated subarrays with identical GPEs are selected, allowing the rotational-invariance relation to be established and the initial DOAs to be estimated from the resulting eigenvalues.

As shown in Fig.~\ref{fig1}, the MDA is divided into
 $\mathcal{P}_{1}=\{1,3,\ldots,M/2-1,M/2+1,\ldots,M-1\}$
and
$\mathcal{P}_{2}=\{2,4,\ldots,M/2,M/2+2,\ldots,M\}$.
The corresponding elements of $\mathcal{P}_{1}$ and $\mathcal{P}_{2}$ are separated by $\frac{\lambda}{2}$. Moreover, according to \eqref{even=odd}, the two subarrays have identical GPEs.
\begin{figure}[h]
	\centering
	\includegraphics[width=1\linewidth]{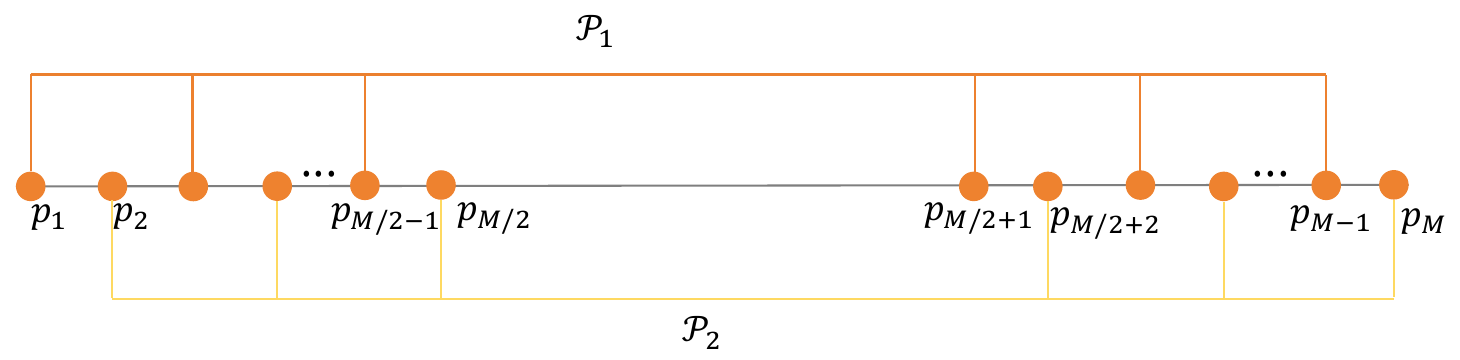}
	\caption{Array geometry for the RIL algorithm.}
	\label{fig1}
\end{figure}

Accordingly, the steering matrices corresponding to the two
subarrays $\mathcal{P}_1$ and $\mathcal{P}_2$ are defined as
\begin{align}
	\bm{A}_1 &= \bm{A}(\mathcal{P}_1,:)
	\in \mathbb{C}^{(M/2)\times K}, \\
	\bm{A}_2 &= \bm{A}(\mathcal{P}_2,:)
	\in \mathbb{C}^{(M/2)\times K},
\end{align}
where $\bm{A}(\mathcal{P}_i,:)$ denotes the steering matrix
constructed from the virtual array-element positions in $\mathcal{P}_i$.

The two subarrays satisfy the rotational-invariance relation
\begin{align}
	\bm{A}_2=\bm{A}_1\bm{\varPhi},
	\label{rotational}
\end{align}
where $\bm{\varPhi}\in\mathbb{C}^{K\times K}$ is defined as
\begin{align}
	\bm{\varPhi}
	=\operatorname{diag}\left(
	e^{j\frac{2\pi}{\lambda}d_1\cos\theta_1},
	e^{j\frac{2\pi}{\lambda}d_1\cos\theta_2},
	\ldots,
	e^{j\frac{2\pi}{\lambda}d_1\cos\theta_K}
	\right)
 . \nonumber
\end{align}

Since $\bm{E}_S$ and $\bm{\varGamma}\bm{A}$ span the same subspace, i.e., $\operatorname{span}(\bm{E}_S)=\operatorname{span}(\bm{\varGamma}\bm{A})$, there exists an invertible matrix $\tilde{\bm{T}}\in\mathbb{C}^{K\times K}$ such that $\bm{E}_S=\bm{\varGamma}\bm{A}\tilde{\bm{T}}$. Based on $\bm{A}_1$ and $\bm{A}_2$, the signal-subspace matrix is partitioned into $\bm{E}_{S_1}=\bm{E}_S(\mathcal{P}_1,:)\in\mathbb{C}^{(M/2)\times K}$ and $\bm{E}_{S_2}=\bm{E}_S(\mathcal{P}_2,:)\in\mathbb{C}^{(M/2)\times K}$, which are expressed as
\begin{align}
	\bm{E}_{S_1}
	&=\bm{\varGamma}_{S,1}\bm{A}_1\tilde{\bm{T}}
	=\tilde{\bm{A}}_1\tilde{\bm{T}},
	\label{E_s1}\\
	\bm{E}_{S_2}
	&=\bm{\varGamma}_{S,2}\bm{A}_2\tilde{\bm{T}}
	=\tilde{\bm{A}}_2\tilde{\bm{T}},
	\label{E_s2}
\end{align}
where $\tilde{\bm{A}}_1
=
\bm{\varGamma}_{S,1}\bm{A}_1
\in\mathbb{C}^{(M/2)\times K}$,
$\tilde{\bm{A}}_2
=
\bm{\varGamma}_{S,2}\bm{A}_2
\in\mathbb{C}^{(M/2)\times K}$. $\bm{\varGamma}_{S,1}=\operatorname{diag}(\bm{\gamma}_{S,1})\in\mathbb{C}^{(M/2)\times(M/2)}$ and $\bm{\varGamma}_{S,2}=\operatorname{diag}(\bm{\gamma}_{S,2})\in\mathbb{C}^{(M/2)\times(M/2)}$, with
\begin{align}
	\bm{\gamma}_{S,1}
	=[&1,\rho_3e^{j\varphi_3},\ldots,
	\rho_{M/2-1}e^{j\varphi_{M/2-1}},1,\nonumber\\
	&\rho_{M/2+3}e^{j\varphi_{M/2+3}},\ldots,
	\rho_{M-1}e^{j\varphi_{M-1}}]^T, \nonumber \\
	\bm{\gamma}_{S,2}
	=[&1,\rho_4e^{j\varphi_4},\ldots,
	\rho_{M/2}e^{j\varphi_{M/2}},1,\nonumber\\
	&\rho_{M/2+4}e^{j\varphi_{M/2+4}},\ldots,
	\rho_Me^{j\varphi_M}]^T. \nonumber
\end{align}
It follows from \eqref{even=odd} that $\bm{\varGamma}_{S,1}=\bm{\varGamma}_{S,2}$ and hence $\tilde{\bm{A}}_2=\tilde{\bm{A}}_1\bm{\varPhi}$. Combining \eqref{E_s1} and \eqref{E_s2} yields
\begin{align}
	\bm{E}_{S_2}=\bm{E}_{S_1}\tilde{\bm{\varPsi}},
	\label{psi}
\end{align}
where $\tilde{\bm{\varPsi}}=\tilde{\bm{T}}^{-1}\bm{\varPhi}\tilde{\bm{T}}\in\mathbb{C}^{K\times K}$ is similar to $\bm{\varPhi}$. The least-squares solution is
\begin{align}
	\tilde{\bm{\varPsi}}
	=\left(\bm{E}_{S_1}^{H}\bm{E}_{S_1}\right)^{-1}
	\bm{E}_{S_1}^{H}\bm{E}_{S_2}.
	\label{Psi}
\end{align}
Therefore, the initial DOA estimates are
\begin{align}
	\hat{\theta}_{k}
	=\arccos\!\left(
	\frac{\arg\!\left(\kappa_k(\tilde{\bm{\varPsi}})\right)\lambda}{2\pi d_1}
	\right),
	\quad k=1,2,\ldots,K,
	\label{DOA0}
\end{align}
where $\kappa_k(\tilde{\bm{\varPsi}})$ is the
eigenvalue of $\tilde{\bm{\varPsi}}$.
The resulting DOA estimates are then sorted in ascending order.

\begin{remark}
	The noise-subspace matrix $\bm{E}_N\in\mathbb{C}^{M\times(M-K)}$ must have a nonzero dimension, requiring
	\begin{align}
		M-K>0.
	\end{align}
	Since $\bm E_{S_1}\in\mathbb C^{(M/2)\times K}$,  the number of sources that can be handled by the
	RIL initialization is dimensionally upper-bounded by
	\begin{align}
		K\leq\frac{M}{2}.
		\label{Kfloor}
	\end{align}
	This upper bound is achievable provided that
	$\bm E_{S_1}$ has full column rank. However,
	\eqref{Kfloor} alone does not guarantee unique recovery.
\end{remark}

The RIL initialization for the MDA with GPEs is summarized in Algorithm~\ref{Algorithm_RIL}. Building on this initialization, the complete AO algorithm is summarized in Algorithm~\ref{Algorithm_AO_EL}.
\begin{algorithm}[t]
	\renewcommand{\algorithmicrequire}{\textbf{Input:}}
	\renewcommand{\algorithmicensure}{\textbf{Output:}}
	\caption{The RIL Initialization Algorithm}
	\label{Algorithm_RIL}
	\begin{algorithmic}[1]
		\Require Signal subspace $\bm{E}_S$,
		subarray index sets $\mathcal{P}_1$ and $\mathcal{P}_2$,
		number of sources $K$,
		displacement $d_1$,
		and wavelength $\lambda$;
		\State Construct
		$\bm{E}_{S_1}=\bm{E}_S(\mathcal{P}_1,:)$
		and
		$\bm{E}_{S_2}=\bm{E}_S(\mathcal{P}_2,:)$; \State Compute $\tilde{\bm{\varPsi}}$
		using \eqref{Psi};
		
		\State Perform the eigendecomposition of
		$\tilde{\bm{\varPsi}}$;
		
		\For{$k=1,2,\ldots,K$}
		\State Compute $\hat{\theta}_k^{(0)}$
		using \eqref{DOA0};
		\EndFor
		
		\State Set
		$\hat{\bm{\theta}}^{(0)}
		=
		[\hat{\theta}_1^{(0)},
		\ldots,
		\hat{\theta}_K^{(0)}]^T$
		and sort its entries;
		
		\Ensure Initial DOA estimates
		$\hat{\bm{\theta}}^{(0)}$;
	\end{algorithmic}
\end{algorithm}
\begin{algorithm}[t]
	\renewcommand{\algorithmicrequire}{\textbf{Input:}}
	\renewcommand{\algorithmicensure}{\textbf{Output:}}
	\caption{The Proposed AO Algorithm with RIL Initialization}
	\label{Algorithm_AO_EL}
	\begin{algorithmic}[1]
		\Require Received signal covariance matrix $\bm{R}$,
		maximum number of iterations $L_{\max}$,
		convergence threshold $\Delta\theta$,
		set of virtual array element positions $\mathcal{P}$,
		number of virtual array elements $M$,
		number of sources $K$,
		inter-element spacing $d$,
		and wavelength $\lambda$;
		
		\State Perform the eigendecomposition of $\bm{R}$ and obtain
		the signal and noise subspaces
		$\bm{E}_S$ and $\hat{\bm{E}}_N$ using \eqref{R};
		
		\State Estimate the noise power $\hat{\sigma}_n^2$ from the
		$M-K$ smallest eigenvalues;
		
		\State Obtain the initial DOA estimates
		$\hat{\bm{\theta}}^{(0)}$
		using the RIL initialization in
		Algorithm~\ref{Algorithm_RIL};
		
		\State Construct $\bm{W}_0$ using \eqref{W1} and \eqref{W2},
		and construct $\bm{b}$ using \eqref{b};
		
		\For{$l=0,1,\ldots,L_{\max}-1$}
		
		\Statex{\textbf{Stage 1: GPE estimation}}
		
		\State Compute $\hat{\bm{\gamma}}^{(l+1)}$
		using \eqref{gamma} with
		$\hat{\bm{\theta}}^{(l)}$ and $\hat{\bm{E}}_N$;
		
		\State Set
		$\hat{\bm{\varGamma}}^{(l+1)}
		=\operatorname{diag}
		(\hat{\bm{\gamma}}^{(l+1)})$;
		
		\Statex{\textbf{Stage 2: DOA estimation}}
		
		\State Calibrate $\bm{R}$ using
		$\hat{\bm{\varGamma}}^{(l+1)}$
		according to \eqref{R_l};
		
		\State Perform the eigendecomposition of the calibrated
		covariance matrix and obtain
		$\tilde{\bm{E}}_N^{(l+1)}$
		using \eqref{E_N_l};
		
		\State Estimate $\hat{\bm{\theta}}^{(l+1)}$
		from the $K$ dominant MUSIC peaks using \eqref{P_l},
		and sort its entries;
		
		\If{$
			\|\hat{\bm{\theta}}^{(l+1)}
			-\hat{\bm{\theta}}^{(l)}\|_2
			< \Delta\theta
			$}
		\State \textbf{break};
		\EndIf
		
		\EndFor
		
		\Ensure Estimated DOAs $\hat{\bm{\theta}}$
		and GPE vector $\hat{\bm{\gamma}}$;
	\end{algorithmic}
\end{algorithm}

 \section{Algorithm Analysis}

\subsection{CRB Analysis}

For the CRB analysis, the source powers and noise variance are assumed to be known. Given $T$ independent samples of a zero-mean Gaussian process, the $(m,n)$-th entry of the Fisher information matrix (FIM) is given by \cite{Kay1993}
\begin{align}
	[\bm{F}] _{m,n}
	&= T \operatorname{tr} \left\{
	\bm{R}^{-1}
	\frac{\partial\bm{R}}{\partial\bm{\eta}_m}
	\bm{R}^{-1}
	\frac{\partial\bm{R}}{\partial\bm{\eta}_n}
	\right\},
	\label{Fmn}
\end{align}
where $\bm{\eta}=[\bm{\theta}^{T},\bm{\rho}^{T},
\bm{\varphi}^{T}]^{T}\in\mathbb R^{(K+M-4)\times1}$, $\bm{\rho}=[\rho_{3},\rho_{5},\ldots,\rho_{M/2-1},
\rho_{M/2+3},\ldots,\rho_{M-1}]^{T}
\in\mathbb{R}^{(M/2-2)\times1}$, $\bm{\varphi}=[\varphi_{3}, \varphi_{5},\ldots,\varphi_{M/2-1}, \varphi_{M/2+3},\ldots,\varphi_{M-1}]^{T}\in\mathbb{R}^{(M/2-2)\times 1}$.

To evaluate the FIM in \eqref{Fmn}, we differentiate the covariance matrix in \eqref{Finalmodel} with respect to $\theta_k$, $\varphi_m$, and $\rho_m$. The derivative with respect to $\rho_m$ is derived in detail in \eqref{R_rho}, while the derivative with respect to $\varphi_m$ is obtained in a similar manner. The resulting derivatives are given in \eqref{R_theta}, \eqref{R_phi} and \eqref{R_rho}.

\begin{align}
	\frac{\partial \bm{R}}{\partial\theta_k}
	&=
	\dot{\bar{\bm{A}}} \bm{\varPi} _{k}\bm{R}_S\bar{\bm{A}}^H
	+
	\bar{\bm{A}}\bm{R}_S \bm{\varPi} _{k}\dot{\bar{\bm{A}}}^H,
	\label{R_theta}\\
	\frac{\partial \bm{R}}{\partial \varphi_m}
	&=
	\bm{G} \bar{\bm{\varPi}} _m\bm{C}_{\varphi}
	+
	\bm{C}_{\varphi}^{H} \bar{\bm{\varPi}} _m\bm{G}^H,
	\label{R_phi}
\end{align}
where
$\bar{\bm{A}}=\bm{G}\bm{D}_{\varphi}\bm{A}\in\mathbb{C}^{M\times K}$ and
$\dot{\bar{\bm{A}}}
=
[\frac{\partial\bar{\bm{a}}(\theta_1)}{\partial\theta_1},\ldots,
\frac{\partial\bar{\bm{a}}(\theta_K)}{\partial\theta_K}]
=
\bm{G}\bm{D}_{\varphi}\dot{\bm{A}}$.
Moreover,
$\bm{\varPi}_k=\bm{i}_k\bm{i}_k^T\in\mathbb{R}^{K\times K}$
denotes the $k$-th diagonal selection matrix,
$\bar{\bm{\varPi}}_m
=
\bm{q}_m\bm{q}_m^T
+
\bm{q}_{m+1}\bm{q}_{m+1}^T
\in\mathbb{R}^{M\times M}$,
$m\in\mathcal{I}_1$.
$\bm{C}_{\rho}
=
\bm{D}_{\varphi}\bm{A}\bm{R}_S\bm{A}^H
\bm{D}_{\varphi}^H\bm{G}^H
\in\mathbb{C}^{M\times M}$
and
$\bm{C}_{\varphi}
=
j\bm{D}_{\varphi}\bm{A}\bm{R}_S\bm{A}^H
\bm{D}_{\varphi}^H\bm{G}^H
\in\mathbb{C}^{M\times M}$.
\begin{figure*}[!t]
	\begingroup
	\begin{align}
		\frac{\partial \bm{R}}{\partial \rho_m}
		=
		\frac{\partial(\bm{G}\bm{D}_{\varphi}\bm{A}\bm{R}_S\bm{A}^H\bm{D}_{\varphi}^H\bm{G}^H)}{\partial \rho_m}
		=
		\frac{\partial \bm{G}}{\partial \rho_m}
		\bm{D}_{\varphi}\bm{A}\bm{R}_S\bm{A}^H\bm{D}_{\varphi}^H\bm{G}^H
		+
		\bm{G}\bm{D}_{\varphi}\bm{A}\bm{R}_S\bm{A}^H\bm{D}_{\varphi}^H
		\frac{\partial \bm{G}^H}{\partial \rho_m}
		=
		 \bar{\bm{\varPi}} _m\bm{C}_{\rho}
		+
		\bm{C}_{\rho}^{H} \bar{\bm{\varPi}} _m.
		\label{R_rho}
	\end{align}
	\endgroup

\end{figure*}

For the FIM submatrix $\bm{F}_{\theta\theta}$,
substituting \eqref{R_theta} into \eqref{Fmn}
yields \eqref{F_theta_theta_detail},
where the third equality follows from $c+c^*=2\mathfrak{R}\{c\}$, while the fourth uses the cyclic property of the trace and $\bm{\varPi}_p=\bm{i}_p\bm{i}_p^H$.
\begin{figure*}[!t]
	\begingroup

	\begin{equation}
		 \begin{aligned}
			[\bm{F}_{\bm{\theta}\bm{\theta}}]_{p,q}
			&=
			T\operatorname{tr}\bigg[
			\bm{R}^{-1}
			\frac{\partial\bm{R}}{\partial\theta_p}
			\bm{R}^{-1}
			\frac{\partial\bm{R}}{\partial\theta_q}
			\bigg]
			=
			T\operatorname{tr}\big[
			\bm{R}^{-1}
			\big(
			\dot{\bar{\bm{A}}}\bm{\varPi}_p\bm{R}_S\bar{\bm{A}}^H
			+
			\bar{\bm{A}}\bm{R}_S\bm{\varPi}_p\dot{\bar{\bm{A}}}^{H}
			\big)
			\bm{R}^{-1}
			\big(
			\dot{\bar{\bm{A}}}\bm{\varPi}_q\bm{R}_S\bar{\bm{A}}^H
			+
			\bar{\bm{A}}\bm{R}_S\bm{\varPi}_q\dot{\bar{\bm{A}}}^{H}
			\big)
			\big]
			\\
			&=
			2T\mathfrak{R}\left\{
			\operatorname{tr}\big[
			\bm{R}^{-1}
			\dot{\bar{\bm{A}}}\bm{\varPi}_p\bm{R}_S\bar{\bm{A}}^H
			\bm{R}^{-1}
			\dot{\bar{\bm{A}}}\bm{\varPi}_q\bm{R}_S\bar{\bm{A}}^H
			\big]
			\right\}
			+
			2T\mathfrak{R}\left\{
			\operatorname{tr}\big[
			\bm{R}^{-1}
			\dot{\bar{\bm{A}}}\bm{\varPi}_p\bm{R}_S\bar{\bm{A}}^H
			\bm{R}^{-1}
			\bar{\bm{A}}\bm{R}_S\bm{\varPi}_q\dot{\bar{\bm{A}}}^{H}
			\big]
			\right\}
			\\
			&=
			2T\mathfrak{R}\left\{
			\bm{i}_p^H\bm{R}_S\bar{\bm{A}}^H
			\bm{R}^{-1}\dot{\bar{\bm{A}}}\bm{i}_q
			\bm{i}_q^H\bm{R}_S\bar{\bm{A}}^H
			\bm{R}^{-1}\dot{\bar{\bm{A}}}\bm{i}_p
			\right\}
			+
			2T\mathfrak{R}\left\{
			\bm{i}_p^H\bm{R}_S\bar{\bm{A}}^H
			\bm{R}^{-1}\bar{\bm{A}}\bm{R}_S\bm{i}_q
			\bm{i}_q^H\dot{\bar{\bm{A}}}^{H}
			\bm{R}^{-1}\dot{\bar{\bm{A}}}\bm{i}_p
			\right\}
			\\
			&=
			2T\mathfrak{R}\left\{
			\big[
			\bm{R}_S\bar{\bm{A}}^H
			\bm{R}^{-1}\dot{\bar{\bm{A}}}
			\big]_{p,q}
			\big[
			\bm{R}_S\bar{\bm{A}}^H
			\bm{R}^{-1}\dot{\bar{\bm{A}}}
			\big]_{q,p}
			\right\}
			+
			2T\mathfrak{R}\left\{
			\big[
			\bm{R}_S\bar{\bm{A}}^H
			\bm{R}^{-1}\bar{\bm{A}}\bm{R}_S
			\big]_{p,q}
			\big[
			\dot{\bar{\bm{A}}}^{H}
			\bm{R}^{-1}\dot{\bar{\bm{A}}}
			\big]_{q,p}
			\right\}.
		\end{aligned} 
		\label{F_theta_theta_detail}
	\end{equation}	\endgroup
	\hrulefill
	\vspace{-2mm}
\end{figure*}
 Accordingly,
\begin{align}
	\bm{F} _{\theta\theta}
	&=
	2T\mathfrak{R}\left\{
	\big(
	\bm{R}_S\bar{\bm{A}}^H\bm{R}^{-1}\dot{\bar{\bm{A}}}
	\big)
	\circledast
	\left(
	\bm{R}_S\bar{\bm{A}}^H\bm{R}^{-1}\dot{\bar{\bm{A}}}
	\right)^T
	\right\}
	\nonumber\\
	&\quad+
	2T\mathfrak{R}\left\{
	\left(
	\bm{R}_S\bar{\bm{A}}^H\bm{R}^{-1}
	\bar{\bm{A}}\bm{R}_S
	\right)
	\circledast
	\left(
	\dot{\bar{\bm{A}}}^H\bm{R}^{-1}
	\dot{\bar{\bm{A}}}
	\right)^T
	\right\}.
\end{align}

The remaining FIM submatrices can be obtained following the same derivation as $\bm{F}_{\theta\theta}$. To compactly represent the joint selection of the two positions associated with each $\bar{\bm{\varPi}}_m$, define the pairwise aggregation matrix
\begin{align}
	\bm{Q}
	=
	\big[
	&\bm{q}_{3}+\bm{q}_{4},
	\bm{q}_{5}+\bm{q}_{6},
	\ldots,
	\bm{q}_{M/2-1}+\bm{q}_{M/2}, \notag\\
	&\bm{q}_{M/2+3}+\bm{q}_{M/2+4},
	\ldots,
	\bm{q}_{M-1}+\bm{q}_{M}
	\big]
	\in
	\mathbb{R}^{M\times(M/2-2)}.
\end{align}
The matrix $\bm{Q}$ is used to pair and sum the two equivalent virtual-element positions corresponding to each physical array element.

The remaining results are directly given below.
\begin{align}
	\bm{F}_{\bm{\theta}\bm{\rho}}
	&=
	2T\mathfrak{R}\Biggl\{
	\bigg[
	\left(
	\bm{R}_S\bar{\bm{A}}^H\bm{R}^{-1}
	\right)
	\circledast
	\left(
	\bm{C}_{\rho}\bm{R}^{-1}
	\dot{\bar{\bm{A}}}
	\right)^T
	\nonumber\\
	&\qquad\quad+
	\left(
	\bm{R}_S\bar{\bm{A}}^H\bm{R}^{-1}
	\bm{C}_{\rho}^H
	\right)
	\circledast
	\left(
	\bm{R}^{-1}\dot{\bar{\bm{A}}}
	\right)^T
	\bigg]\bm{Q}
	\Biggr\},
	\label{F_theta_rho}\\[1ex]
	\bm{F}_{\bm{\theta}\bm{\varphi}}
	&=
	2T\mathfrak{R}\Biggl\{
	\bigg[
	\left(
	\bm{R}_S\bar{\bm{A}}^H\bm{R}^{-1}\bm{G}
	\right)
	\circledast
	\left(
	\bm{C}_{\varphi}\bm{R}^{-1}
	\dot{\bar{\bm{A}}}
	\right)^T
	\nonumber\\
	&\qquad\quad+
	\left(
	\bm{R}_S\bar{\bm{A}}^H\bm{R}^{-1}
	\bm{C}_{\varphi}^H
	\right)
	\circledast
	\left(
	\bm{G}^H\bm{R}^{-1}
	\dot{\bar{\bm{A}}}
	\right)^T
	\bigg]\bm{Q}
	\Biggr\},
	\label{F_theta_phi}\\[1ex]
	\bm{F}_{\bm{\rho}\bm{\rho}}
	&=
	2T\mathfrak{R}\Biggl\{
	\bm{Q}^T
	\bigg[
	\left(
	\bm{C}_{\rho}\bm{R}^{-1}
	\right)
	\circledast
	\left(
	\bm{C}_{\rho}\bm{R}^{-1}
	\right)^T
	\nonumber\\
	&\qquad\quad+
	\left(
	\bm{C}_{\rho}\bm{R}^{-1}\bm{C}_{\rho}^H
	\right)
	\circledast
	\left(
	\bm{R}^{-1}
	\right)^T
	\bigg]
	\bm{Q}
	\Biggr\},
	\label{F_rho_rho}\\[1ex]
	\bm{F}_{\bm{\rho}\bm{\varphi}}
	&=
	2T\mathfrak{R}\Biggl\{
	\bm{Q}^T
	\bigg[
	\left(
	\bm{C}_{\rho}\bm{R}^{-1}\bm{G}
	\right)
	\circledast
	\left(
	\bm{C}_{\varphi}\bm{R}^{-1}
	\right)^T
	\nonumber\\
	&\qquad\quad+
	\left(
	\bm{C}_{\rho}\bm{R}^{-1}
	\bm{C}_{\varphi}^H
	\right)
	\circledast
	\left(
	\bm{G}^H\bm{R}^{-1}
	\right)^T
	\bigg]
	\bm{Q}
	\Biggr\},
	\label{F_rho_phi}\\[1ex]
	\bm{F}_{\bm{\varphi}\bm{\varphi}}
	&=
	2T\mathfrak{R}\Biggl\{
	\bm{Q}^T
	\bigg[
	\left(
	\bm{C}_{\varphi}\bm{R}^{-1}\bm{G}
	\right)
	\circledast
	\left(
	\bm{C}_{\varphi}\bm{R}^{-1}\bm{G}
	\right)^T
	\nonumber\\
	&\qquad\quad+
	\left(
	\bm{C}_{\varphi}\bm{R}^{-1}
	\bm{C}_{\varphi}^H
	\right)
	\circledast
	\left(
	\bm{G}^H\bm{R}^{-1}\bm{G}
	\right)^T
	\bigg]
	\bm{Q}
	\Biggr\}.
	\label{F_phi_phi}
\end{align}

Collecting the above six FIM submatrices, the complete FIM
associated with $\bm{\eta}=[\bm{\theta}^{T},
\bm{\rho}^{T},\bm{\varphi}^{T}]^{T}$ is expressed as
\begin{equation} 
	\bm{F}= 
	\begin{bmatrix} 
		\bm{F}_{\bm{\theta}\bm{\theta}} 
		& \bm{F}_{\bm{\theta}\bm{\rho}} 
		& \bm{F}_{\bm{\theta}\bm{\varphi}} 
		\\ 
		\bm{F}_{\bm{\theta}\bm{\rho}}^{T} 
		& \bm{F}_{\bm{\rho}\bm{\rho}} 
		& \bm{F}_{\bm{\rho}\bm{\varphi}} 
		\\ 
		\bm{F}_{\bm{\theta}\bm{\varphi}}^{T} 
		& \bm{F}_{\bm{\rho}\bm{\varphi}}^{T} 
		& \bm{F}_{\bm{\varphi}\bm{\varphi}} 
	\end{bmatrix}. 
	\label{complete_FIM} 
\end{equation}
  Therefore, the CRB matrix of \(\bm{\eta}\) is given by
\begin{equation}
	\bm{\mathrm{CRB}}(\bm{\eta})=\bm{F}^{-1}.
	\label{joint_CRB}
\end{equation}
Since \(\bm{\eta}\) incorporates all the unknown parameters, the CRB submatrices corresponding to the DOAs, gain errors, and phase errors can be extracted from \(\bm{F}^{-1}\). To facilitate comparison with the root mean-square error (RMSE) results in the simulations, their scalar CRBs are respectively defined as

\begin{align}
	\mathrm{CRB}(\bm{\theta})
	&=
	 \sqrt{
		\frac{1}{K}
		\operatorname{tr}\!\left(
		[\bm{\mathrm{CRB}}(\bm{\eta})]_{
			\mathcal{J}_{\theta},\mathcal{J}_{\theta}}
		\right)
	},
	\label{CRB_theta}
	\\
	\mathrm{CRB}(\bm{\rho})
	&=
	 \sqrt{
		\frac{1}{\frac{M}{2}-2}
		\operatorname{tr}\!\left(
		[\bm{\mathrm{CRB}}(\bm{\eta})]_{
			\mathcal{J}_{\rho},\mathcal{J}_{\rho}}
		\right)
	},
	\label{CRB_rho}
	\\
	\mathrm{CRB}(\bm{\varphi})
	&=
	 \sqrt{
		\frac{1}{\frac{M}{2}-2}
		\operatorname{tr}\!\left(
		[\bm{\mathrm{CRB}}(\bm{\eta})]_{
			\mathcal{J}_{\varphi},\mathcal{J}_{\varphi}}
		\right)
	},
	\label{CRB_varphi}
\end{align}
where 
\begin{align}
	\mathcal{J}_{\theta}
	&=\{1,\ldots,K\},\nonumber\\
	\mathcal{J}_{\rho}
	&=\left\{K+1,\ldots,K+\frac{M}{2}-2\right\},\nonumber\\
	\mathcal{J}_{\varphi}
	&=\left\{K+\frac{M}{2}-1,\ldots,K+M-4\right\}.\nonumber
\end{align}

	\subsection{Computational Complexity Analysis}
The computational complexity of the proposed algorithm is analyzed as follows. 
Let $N_{\theta}$ denote the number of search-grid points for MUSIC and 
 $I_{\mathrm{iter}}$ denote the average number of AO iterations required for 
convergence. For the RIL initialization, the complexities of constructing the 
sample covariance matrix and performing its eigendecomposition are 
$\mathcal{O}(M^{2}T)$ and $\mathcal{O}(M^{3})$, respectively. Moreover, the 
least-squares calculation in \eqref{Psi} and the subsequent eigendecomposition for 
initial DOA estimation entail $\mathcal{O}(MK^{2}+K^{3})$ operations. Thus, 
the computational complexity of the RIL initialization is
$\mathcal{O}(M^{2}T+M^{3}+MK^{2}+K^{3})$.

For each iteration of Algorithm~\ref{Algorithm_AO_EL}, the computational complexity of 
Stage~1
 mainly comes from the construction of 
$\hat{\bm E}_{N}\hat{\bm E}_{N}^{H}$, the calculation of $\hat{\bm V}$, and 
the matrix inversions in \eqref{gamma}, whose complexities are 
$\mathcal{O}(M^{2}(M-K))$, $\mathcal{O}(KM^{2})$, and $\mathcal{O}(M^{3})$, 
respectively. In Stage~2, the covariance-matrix calibration in 
\eqref{R_l} requires $\mathcal{O}(M^{2})$ operations, while the 
eigendecomposition in \eqref{R_l_evd} and the MUSIC spectral search in 
\eqref{P_l} have complexities of $\mathcal{O}(M^{3})$ and 
$\mathcal{O}(N_{\theta}M(M-K))$, respectively. Therefore, the overall 
computational complexity of the proposed algorithm is
\begin{equation}
	 \begin{aligned}
		C_{\mathrm{total}}
		=&~\mathcal{O}\!\left(M^{2}T+M^{3}+MK^{2}+K^{3}\right) \\
		&+I_{\mathrm{iter}}
		\mathcal{O}\!\left(M^{3}+KM^{2}+N_{\theta}M(M-K)\right).
	\end{aligned} 
	\label{eq:computational_complexity}
\end{equation}
Since $K\leq  M/2$ for the RIL initialization, 
the above complexity can be further simplified as
 $\mathcal{O}(M^{2}T+M^{3})
+I_{\mathrm{iter}}\mathcal{O}(M^{3}+N_{\theta}M^{2})$.
In particular, when $N_{\theta}\gg M$, the MUSIC spectral search dominates the 
per-iteration computational complexity.

\section{Numerical Simulations}

In this section, numerical simulations are conducted to evaluate the performance of the proposed self-calibration AO algorithm. Unless otherwise specified, the number of virtual array antennas is $M=16$, the equivalent-array aperture is $H=1.5(M-1)d$, the source DOAs are $(30^\circ,60^\circ,135^\circ)$, the signal-to-noise ratio (SNR) is $10~\mathrm{dB}$ and the number of snapshots is $T=100$.  For the MUSIC algorithm, the coarse and fine search step sizes are set to $1^\circ$ and $0.1^\circ$. In addition, the convergence threshold $\Delta\theta$ and the maximum number of iterations for the AO algorithm are set to $10^{-4}$ and $30$, respectively. The SNR is defined as
\begin{align}
	\mathrm{SNR}
	=
	10\lg
	\left(
	\frac{\sigma_k^2}{\sigma_n^2}
	\right).
\end{align}

To evaluate the performance of the proposed algorithms, we use RMSE to describe the estimation error and the RMSE of the DOA estimation is defined as
\begin{align}
	\mathrm{RMSE}
	=
	 \sqrt{
		\frac{1}{KX}
		\sum_{r=1}^{X}
		\sum_{k=1}^{K}
		\left(
		\hat{\theta}_{r,k} -\theta_k
		\right)^2
	},
\end{align}
where $X$ denotes the number of Monte Carlo trials and $X$ is set to $300$, $\hat{\theta}_{r,k} $ denotes the estimate of the $k$-th DOA in the  $r$-th Monte Carlo trial, and $\theta_k$ denotes the corresponding true DOA.

Following \cite{YeRobust}, the gain and phase errors are generated as
\begin{align}
	\rho
	&=
	1+\sqrt{12}\sigma_\rho\varepsilon,
	\label{rho}\\
	\varphi
	&=
	\sqrt{12}\sigma_\varphi\tau,
	\label{varphi}
\end{align}
where $\varepsilon$ and $\tau$ are independently and uniformly distributed in $[-0.5,0.5]$, and $\sigma_\rho$ and $\sigma_\varphi$ denote the standard deviations of the gain and phase errors, respectively.  Unless otherwise specified, $\sigma_\varphi=40^\circ$ and $\sigma_\rho=0.2$.

The DOA estimation in the $r$-th Monte Carlo trial is considered successful  if the estimation error
of each source is less than $1^\circ$. The set of successful trials and the corresponding success probability $P_{\mathrm{s}}$ are defined as
\begin{align}
	&\mathcal{S}
	=
	\Bigl\{
	 r \in\{1,\ldots, X \}
	\mid
	\max_{1\leq k\leq K}
	\bigl|
	\hat{\theta} _{r,k} -\theta_k
	\bigr|
	<\delta_\theta
	\Bigr\},\delta_\theta=1^\circ 
	\label{success_set}\\
	&P_{\mathrm{s}}
	=
	 \frac{\lvert\mathcal{S}\rvert}{X}
	,
	\label{success_probability}
\end{align}
where $\hat{\theta}_{r,k}$ represents the angle estimate of the $k$-th source in the  $r$-th Monte Carlo trial, $\delta_\theta=1^\circ$ denotes the success threshold and $\lvert\mathcal{S}\rvert$ denotes the number of successful trials.

In the following simulations, the proposed AO algorithm with RIL initialization is referred to as \textbf{Baseline 1}. For performance comparison, we consider the following baseline schemes:
\begin{itemize}
	\item \textbf{Baseline 2}: the AO algorithm with random initialization, where $K$ initial angles are randomly generated over  $(0^\circ,180^\circ)$.
	\item \textbf{Baseline 3}: the proposed RIL algorithm applied to the MDA in the presence of GPEs.
	\item \textbf{Baseline 4}: conventional one-dimensional MUSIC applied to the MDA in the presence of GPEs. The GPEs are not calibrated.

\end{itemize}
In addition, the CRB under the considered MDA configuration  is included as a theoretical performance benchmark.

To illustrate the GPE estimation accuracy of the proposed AO algorithm, Fig.~\ref{fig:verify_parameters} compares the estimated gain coefficients and phase errors with their corresponding true values for the odd-indexed elements of the virtual array, respectively. For the virtual elements associated with the calibrated first antenna, the gain coefficients are set to unity and the phase errors to zero. Specifically, the virtual elements indexed by $1,3,\ldots,15$ are considered. The gain coefficients $\rho$ and phase errors $\varphi$ are randomly generated according to \eqref{rho} and \eqref{varphi}, respectively. It can be observed that the gain-coefficient estimates obtained by the proposed algorithm closely follow the true values, while its phase-error estimates nearly coincide with the true values. The results demonstrate that the proposed AO algorithm can accurately estimate the GPEs. By contrast, the AO algorithm with random initialization exhibits relatively large estimation errors for both $\rho$ and $\varphi$, which also provides preliminary evidence of the importance of reliable initial DOA estimates for the AO algorithm.

\begin{figure}[!t]
	\centering
	\begin{subfigure}[t]{0.48\linewidth}
		\centering
		\includegraphics[width=\linewidth]{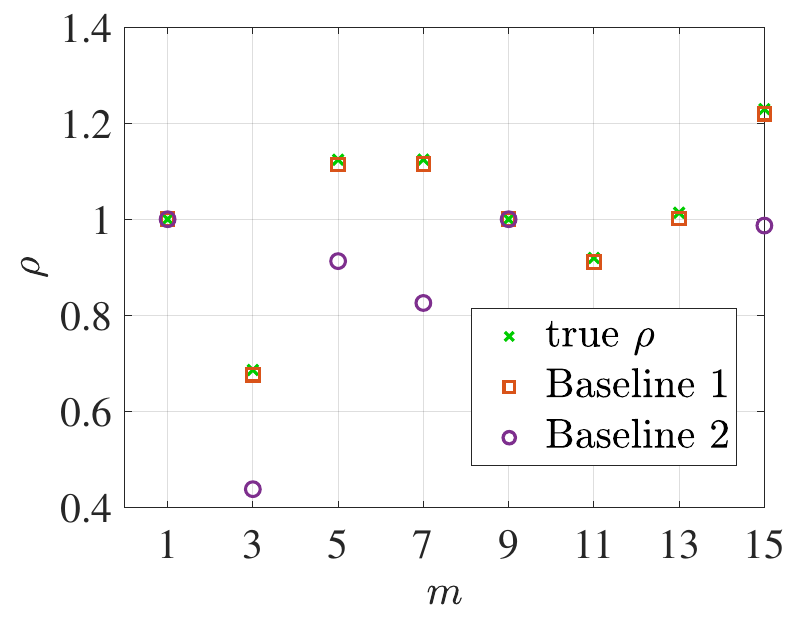}
		\caption{Estimated and true values of $\rho$.}
		\label{fig:verify_rho}
	\end{subfigure}
	\hfill
	\begin{subfigure}[t]{0.48\linewidth}
		\centering
		\includegraphics[width=\linewidth]{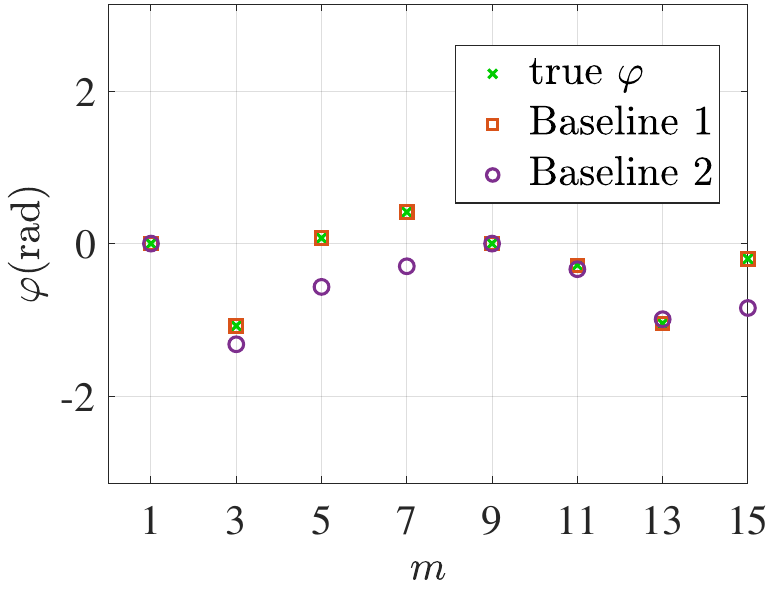}
		\caption{Estimated and true values of $\varphi$.}
		\label{fig:verify_phi}
	\end{subfigure}
	\caption{Verification of the estimated GPEs.}
	\label{fig:verify_parameters}
\end{figure}

To illustrate the convergence behavior of the proposed AO algorithm, Fig.~\ref{theta_RMSE_Iter} presents the DOA estimation RMSE versus the iteration index at SNRs of $10$ dB and $20$ dB. For the AO algorithm with random initialization, the RMSE changes only slightly with the iteration index and remains relatively high at both SNRs. This is because the initial DOAs provided by the AO algorithm with random initialization deviate substantially from their true values. Consequently, the first stage produces inaccurate GPE estimates, and the covariance matrix cannot be well calibrated. The resulting DOA estimates still contain large errors and adversely affect GPE estimation in the next iteration. Therefore, the RMSE of the AO algorithm with random initialization remains high throughout the iterative process.
By contrast, the RMSE of the proposed algorithm gradually decreases with the iteration index and eventually becomes stable. The RMSE at $20$ dB is consistently lower than that at $10$ dB. These results show that the proposed iterative algorithm progressively improves the DOA estimation accuracy.  This is because RIL first provides relatively accurate initial DOA estimates. The first stage can therefore obtain reliable GPE estimates and then we can use them to calibrate the covariance matrix in stage 2. The calibrated covariance matrix produces a more accurate noise-subspace estimate and improves the subsequent DOA estimation. The updated DOAs are then used to refine the GPE estimates in the next iteration. Through this alternating procedure, the DOA and GPE estimation errors gradually decrease. Moreover, all curves of the proposed algorithm stabilize within approximately 15 iterations. Thus, the proposed algorithm requires only a limited number of iterations to converge.
\begin{figure}[!t]
	\centering
	\includegraphics[width=0.85\linewidth]{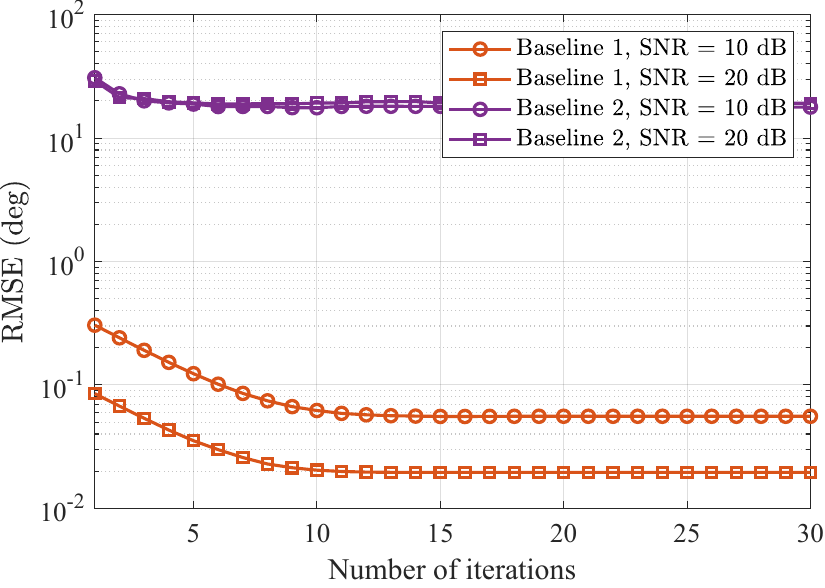}
	\caption{DOA estimation RMSE versus the iteration index.}
	\label{theta_RMSE_Iter}
\end{figure}

\begin{figure}[!t]
	\centering
	\includegraphics[width=0.85\linewidth]{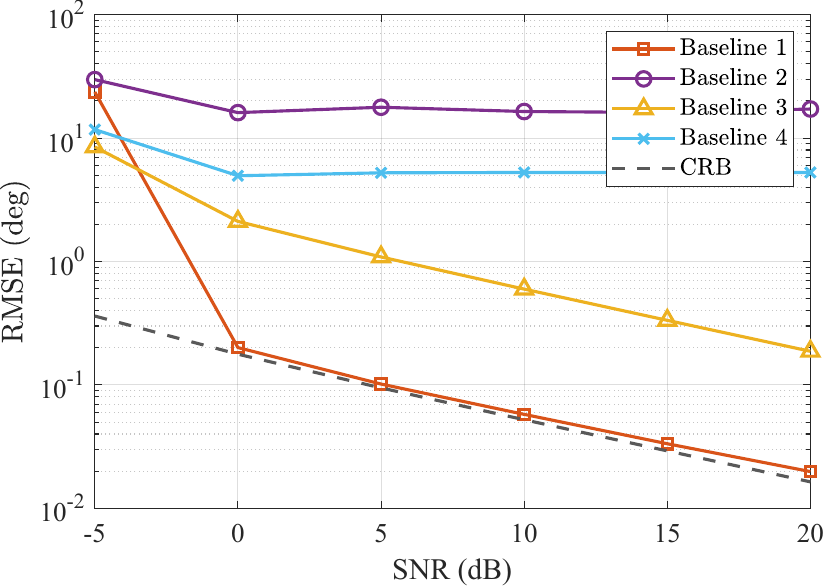}
	\caption{DOA estimation RMSE versus SNR.}
	\label{theta_SNR}
\end{figure}

Fig.~\ref{theta_SNR} depicts the DOA estimation RMSE of different methods versus SNR, where  the SNR takes values in $\{-5,0,5,10,15,20\}~\mathrm{dB}$. As the SNR increases, the RMSE of the proposed algorithm decreases significantly and gradually approaches the corresponding CRB. The proposed algorithm also achieves a lower RMSE than  the AO algorithm with random initialization over the entire SNR range. By contrast, MUSIC exhibits a pronounced performance floor because it does not compensate for the GPEs.
Besides, at an SNR of $-5$ dB, the proposed algorithm performs worse than RIL. The main reason is that RIL cannot provide sufficiently accurate initial DOA estimates under severe noise contamination. Consequently, the subsequent GPE estimation and covariance-matrix calibration cannot effectively improve the DOA estimates. As the SNR increases, the RIL estimates become more reliable, and the performance gap between the proposed algorithm and RIL gradually narrows. The two curves intersect between $-5$ dB and $0$ dB. At SNRs of $5$ dB and above, the proposed algorithm outperforms RIL. This result is consistent with the initialization sensitivity observed in Fig.~\ref{theta_RMSE_Iter}.

\begin{figure}[!t]
	\centering
	\includegraphics[width=0.85\linewidth]{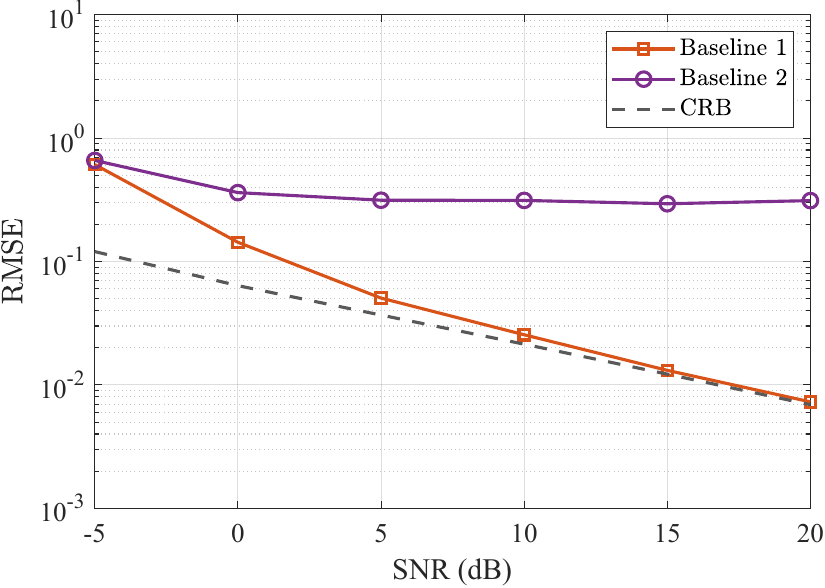}
	\caption{Gain estimation RMSE versus SNR.}
	\label{G_SNR}
\end{figure}
\begin{figure}[!t]
	\centering
	\includegraphics[width=0.85\linewidth]{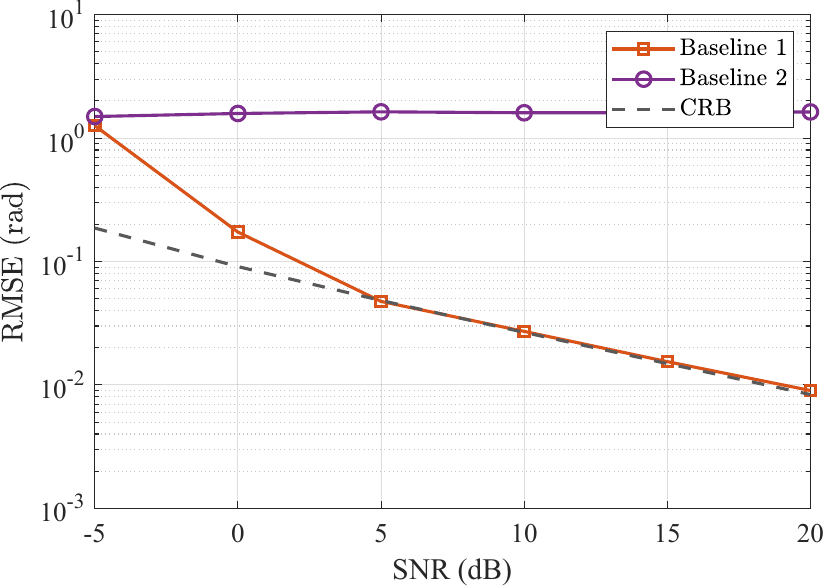}
	\caption{Phase estimation RMSE versus SNR.}
	\label{P_SNR}
\end{figure}
In order to evaluate the GPE estimation performance of the proposed AO algorithm, Fig.~\ref{G_SNR}  and Fig.~\ref{P_SNR} depict the RMSEs of the estimated gain coefficients $\rho$ and phase errors $\varphi$ versus SNR. In this simulation, the SNR takes values in $\{-5,0,5,10,15,20\}~\mathrm{dB}$. The RMSEs of GPEs are computed over the $M-4$ uncalibrated virtual elements, excluding the four calibrated reference elements.
For the proposed algorithm, the RMSEs of both $\rho$ and $\varphi$ decrease significantly as the SNR increases. They are also substantially lower than the corresponding RMSEs of the  AO algorithm with random initialization. By contrast, both RMSEs of the AO algorithm with random initialization remain relatively high and change only slightly with SNR.
Moreover, the gain and phase RMSEs of the proposed algorithm gradually approach their corresponding CRBs at medium and high SNRs. These results demonstrate that the proposed AO algorithm can effectively recover the unknown GPEs. The accurate GPE estimates also explain why the DOA estimation performance of the proposed algorithm approaches the CRB in Fig.~\ref{theta_SNR}. As $\hat{\bm{\varGamma}}$ approaches the true $\bm{\varGamma}$, the calibrated covariance matrix becomes closer to the ideal  covariance matrix. This reduces the array-manifold mismatch and improves the orthogonality between the steering vectors and the noise subspace. Therefore, more accurate GPE estimates directly improve the subsequent DOA estimation performance.

Fig.~\ref{theta_RMSE_H} presents the DOA estimation RMSE of different methods versus the equivalent-array aperture $H$, where   $H$ increases from $15\lambda/2$ to $15\lambda$ in increments of $\lambda/2$. In particular, when $H=15\lambda/2$, the MDA reduces to a ULA.
It can be observed that the RMSE of the proposed algorithm is the lowest over the entire range of $H$ and generally declines as $H$ increases. This is because the enlarged equivalent-array aperture improves the angular resolution. By iteratively estimating the GPEs and calibrating the covariance matrix, the proposed algorithm can effectively benefit from the enlarged aperture.
By contrast,  the AO algorithm with random initialization exhibits much larger and strongly fluctuating RMSEs. Meanwhile, MUSIC shows no consistent improvement as $H$ increases and its RMSE rises markedly at larger apertures because of the array-manifold mismatch caused by uncompensated GPEs. In addition, the RMSE of RIL is consistently lower than that of MUSIC but remains nearly unchanged over the considered range of $H$. This is because RIL only exploits the rotational-invariance relation to obtain the initial DOA estimates. It does not explicitly estimate all the GPEs or calibrate the complete covariance matrix. Therefore, increasing $H$ provides only limited improvement in the estimation accuracy of RIL.

\begin{figure}[!t]
	\centering
	\includegraphics[width=0.85\linewidth]{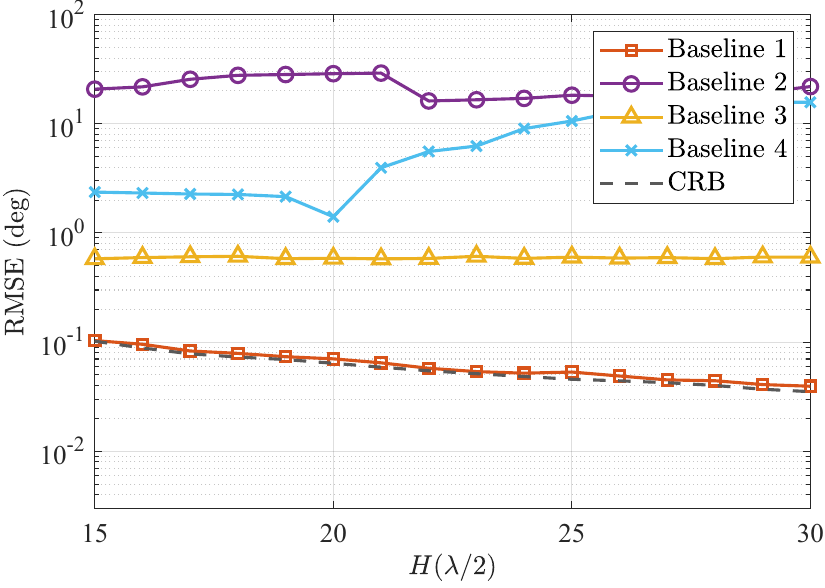}
	\caption{DOA estimation RMSE versus the equivalent-array aperture.}
	\label{theta_RMSE_H}
\end{figure}
\begin{figure}[!t]
	\centering
	\includegraphics[width=0.85\linewidth]{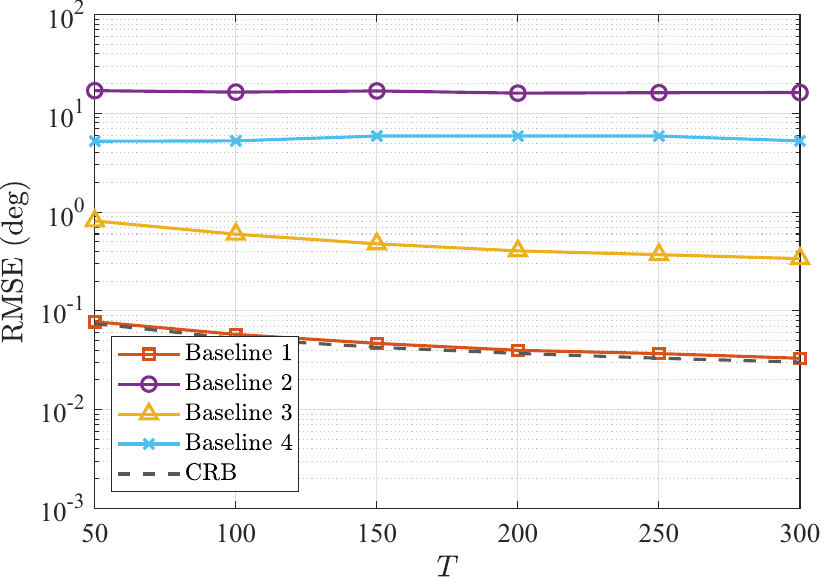}
	\caption{DOA estimation RMSE versus the number of snapshots.}
	\label{theta_RMSE_T}
\end{figure}
Fig.~\ref{theta_RMSE_T} compares the DOA estimation RMSEs of different methods as the number of snapshots $T$ increases from $50$ to $300$ in increments of $50$. The proposed AO algorithm consistently achieves the lowest RMSE and closely follows the CRB, with its performance steadily improving as $T$ increases. This is because more snapshots improve the estimation accuracy of the sample covariance matrix, signal and noise subspaces, and GPEs. RIL also benefits from increasing $T$ and consistently achieves a substantially lower RMSE than MUSIC. In contrast, MUSIC remains at a relatively high RMSE level with only minor fluctuations, since additional snapshots cannot eliminate the array-manifold mismatch caused by unknown GPEs. Similarly, the AO algorithm with random initialization shows little improvement as $T$ increases, indicating that improved covariance estimation alone cannot compensate for the performance degradation caused by unreliable initialization.

\begin{figure}[!t]
	\centering
	\includegraphics[width=0.85\linewidth]{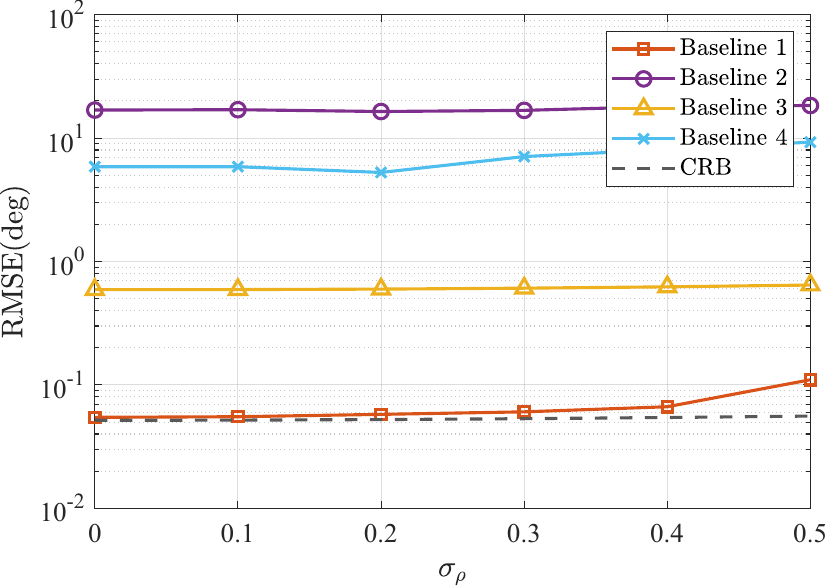}
	\caption{DOA estimation RMSE versus $\sigma_\rho$.}
	\label{RMSE_g}
\end{figure}
Fig.~\ref{RMSE_g} presents the DOA estimation RMSE versus the gain-error standard deviation $\sigma_\rho$, where $\sigma_\varphi$ is set to $40^\circ$.
As $\sigma_\rho$ increases from $0$ to $0.4$, the RMSE of the proposed algorithm increases only moderately. A more noticeable degradation occurs at $\sigma_\rho=0.5$. Nevertheless, the proposed algorithm still achieves the lowest RMSE among all the considered methods. This demonstrates that the iterative GPE estimation and covariance matrix calibration effectively mitigate the effect of gain errors on DOA estimation. By contrast, the AO algorithm with random initialization remains at a considerably higher RMSE level over the entire range, further highlighting the importance of reliable initialization for the subsequent GPE estimation and array calibration.
RIL also maintains a substantially lower RMSE than MUSIC over the entire considered range of $\sigma_\rho$. Since a phase-error standard deviation of $40^\circ$ is present even when $\sigma_\rho=0$, MUSIC already suffers from considerable array-manifold mismatch at the starting point. Its RMSE generally increases as the gain-error level becomes larger, although minor nonmonotonic fluctuations can be observed.

Fig.~\ref{RMSE_phi} depicts the DOA estimation RMSE versus the phase-error standard deviation $\sigma_\varphi$, where $\sigma_\rho=0.2$. The value of $\sigma_\varphi$ increases from $0^\circ$ to $60^\circ$ in increments of $10^\circ$.
It can be observed that the RMSE of the proposed algorithm remains nearly unchanged as $\sigma_\varphi$ increases and is the lowest at all considered nonzero phase-error levels. At $\sigma_\varphi=0^\circ$, MUSIC performs slightly better, possibly because GPE estimation and covariance-matrix calibration introduce small additional errors with a finite number of snapshots. As $\sigma_\varphi$ increases, uncompensated phase errors cause array-manifold mismatch, and the RMSE of MUSIC rapidly exceeds that of the proposed algorithm.
When $\sigma_\varphi\leq20^\circ$, MUSIC also achieves a lower RMSE than RIL. This is because the array-manifold mismatch caused by the phase errors remains limited, while MUSIC exploits the full-array observations. As $\sigma_\varphi$ further increases, the performance of MUSIC degrades rapidly, whereas the RMSE of RIL remains nearly unchanged. This result demonstrates the robustness of RIL against phase errors. The RMSE of the AO algorithm with random initialization remains relatively high over the entire considered range, indicating that phase errors are not the main factor limiting its performance. Instead, its performance is mainly limited by unreliable initialization.

 \begin{figure}[!t]
	\centering
	\includegraphics[width=0.85\linewidth]{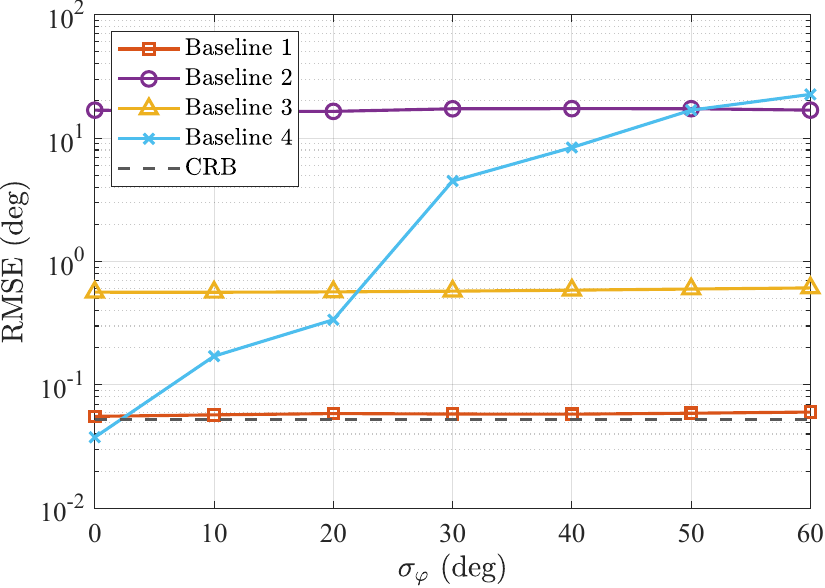}
	\caption{DOA estimation RMSE versus $\sigma_\varphi$.}
	\label{RMSE_phi}
\end{figure}
Fig.~\ref{sucess_K} shows the DOA estimation success probability versus the number of sources. The simulation parameters are set to $\mathrm{SNR}=10~\mathrm{dB}$ and $K$ varies from $2$ to $7$ with a step of one and the DOAs are given by
 $\bm{\theta}_s = \left[ 30^\circ,\, 40^\circ,\, \dots,\, 30^\circ + 10^\circ(K - 1) \right]^T$. For $M=16$, the considered values of $K$ satisfy the upper bound
 in \eqref{Kfloor}.
We can observe that MUSIC applied to the MDA almost completely fails in the presence of GPEs when $K\ge 3$. Moreover, the success probability of RIL decreases significantly as $K$ increases. This is because a larger number of sources increases the dimension of the signal subspace while reducing that of the noise subspace, making the subspace estimation and the separation of multiple DOAs more difficult.
In contrast, the success probability of the proposed algorithm remains close to $100\%$ over the entire considered range of $K$. This demonstrates that the proposed AO refinement can significantly improve the accuracy and reliability of multiple-source DOA estimation. However, the success probability of the AO algorithm with random initialization remains at a relatively low level.

\begin{figure}[!t]
	\centering
	\includegraphics[width=0.85\linewidth]{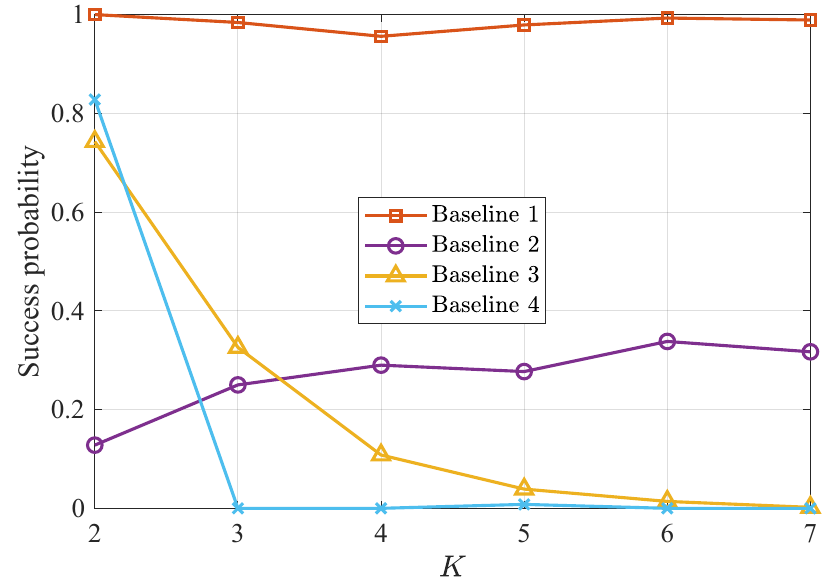}
	\caption{DOA estimation success probability versus $K$.}
	\label{sucess_K}
\end{figure}

	\section{Conclusion}

In this paper, we proposed a two-stage self-calibration DOA estimation method for MDA systems affected by unknown GPEs. By exploiting the structural relationships introduced by array motion, the proposed method enables joint GPE calibration and DOA refinement without requiring prior knowledge of the unknown array errors. The proposed AO framework jointly refines the GPE and DOA estimates, while the RIL scheme provides reliable initial estimates for the iterative procedure. Theoretical analysis was further supported by the derived CRB. Simulation results showed that the proposed algorithm provides accurate GPE estimation, reliable convergence, and robust DOA estimation over a wide range of system conditions. In particular, its performance improves with increasing SNR and approaches the CRB, while conventional MUSIC suffers significant degradation in the presence of uncompensated GPEs. These results demonstrate the effectiveness of the proposed self-calibration framework for MDA-based DOA estimation under practical array imperfections.
\appendices
\section{Proof of  Theorem~\ref{thm:theorem1}}
\label{app:proof_theorem1}
In this appendix, we prove  Theorem~\ref{thm:theorem1} by deriving the closed-form
solution to problem (P3) using the Lagrange multiplier method.
\begin{proof}
Let $\mathcal{L}(\bm{\gamma},\bm{\mu})$ denote the Lagrangian. Since $\mathcal{L}(\bm{\gamma},\bm{\mu})$ is real-valued, it can be constructed as
\begin{align} 
	\mathcal L(\bm\gamma,\bm\mu) &=\bm\gamma^H\hat{\bm V} \bm\gamma+2 \mathfrak{R} \left\{ \bm\mu^H\left(\bm W_0^H\bm\gamma-\bm b\right) \right\}\nonumber\\
	&=\bm\gamma^H\hat{\bm V} \bm\gamma +\bm\mu^H(\bm W_0^H\bm\gamma-\bm b)+(\bm\gamma^H\bm W_0-\bm b^H)\bm\mu, \label{langrange}
\end{align}
where $\bm{\mu}$ is the Lagrange multiplier vector. 

Taking the Wirtinger derivative with respect to $\bm\gamma^*$ gives
\begin{align}
	\frac{\partial \mathcal L}{\partial\bm\gamma^*} = \hat{\bm V} \bm\gamma+\bm W_0\bm\mu.
\end{align}

Setting it to zero yields
\begin{align}		
	\hat{\bm\gamma}  = -\hat{\bm V} ^{-1} \bm W_0\bm\mu. \label{gamma hat}
\end{align}

Since $\bm{\mu}$ in \eqref{gamma hat} is unknown, substituting \eqref{gamma hat} into the constraint $\bm{W}_0^{H}\bm{\gamma}=\bm{b}$ gives
\begin{align}
	\bm{\mu}=- \left(\bm{W}_0^{H}\hat{\bm{V}} ^{-1}\bm{W}_0 \right) ^{-1}\bm{b}.\label{mu}
\end{align}

Substituting \eqref{mu} into \eqref{gamma hat} yields
\begin{align}
	\hat{\bm{\gamma}}  =  \hat{\bm{V}}  ^{-1} \bm{W}_{0} \left(\bm{W}_{0}^{H} \hat{\bm{V}} ^{-1} \bm{W}_{0} \right) ^{-1} \bm{b}, 
\end{align}
which is identical to the closed-form solution in \eqref{gamma}.

This completes the proof of  Theorem~\ref{thm:theorem1}.
\end{proof}

\FloatBarrier
\raggedbottom 
\bibliographystyle{IEEEtran}
\bibliography{mybib_Abbreviation_1}

@ARTICLE{ZhangRYISAC,
	author={Zhou, Wenxing and Zhang, Ruoyu and Chen, Guangyi and Wu, Wen},
	journal={IEEE Open J. Commun. Soc.}, 
	title={Integrated sensing and communication waveform design: a survey}, 
	year={2022},
	volume={3},
	number={},
	pages={1930-1949},
	doi={10.1109/OJCOMS.2022.3215683}}

@ARTICLE{ZhangRYDOA,
	author={Zhang, Ruoyu and Shim, Byonghyo and Wu, Wen},
	journal={IEEE Trans. Signal Process.}, 
	title={Direction-of-arrival estimation for large antenna arrays with hybrid analog and digital architectures}, 
	year={2022},
	volume={70},
	number={},
	pages={72-88},
	doi={10.1109/TSP.2021.3119768}}

@ARTICLE{JiangMAAssisted,
	author={Liu, Rang and Li, Ming and Liu, Qian and Swindlehurst, A. Lee},
	journal={IEEE Trans. Wireless Commun.}, 
	title={{DOA} estimation-oriented joint array partitioning and beamforming designs for {ISAC} systems}, 
	year={2025},
	volume={24},
	number={3},
	pages={2052-2066},
	doi={10.1109/TWC.2024.3516037}}

@ARTICLE{ZhangRYMIMO,
	author={Zhang, Ruoyu and Cheng, Lei and Zhang, Wei and Guan, Xinrong and Cai, Yueming and Wu, Wen and Zhang, Rui},
	journal={IEEE Wireless Commun. Lett.}, 
	title={Channel estimation for movable-antenna {MIMO} systems via tensor decomposition}, 
	year={2024},
	volume={13},
	number={11},
	pages={3089-3093},
	doi={10.1109/LWC.2024.3450592}}

@INPROCEEDINGS{WFAO,
	author={Friedlander, B. and Weiss, A. J. },
	booktitle={Proc. IEEE Int. Conf. Acoust., Speech Signal Process. (ICASSP)}, 
	title={Eigenstructure methods for direction finding with sensor gain and phase uncertainties}, 
	year={1988},
	volume={5},
	number={},
	pages={2681-2684 },
	doi={10.1109/ICASSP.1988.197201}}

@ARTICLE{LiuAFWF,
	author={Liu, Aifei and Liao, Guisheng and Zeng, Cao and Yang, Zhiwei and Xu, Qing},
	journal={IEEE Trans. Signal Process.}, 
	title={An eigenstructure method for estimating {DOA} and sensor gain-phase errors}, 
	year={2011},
	volume={59},
	number={12},
	pages={5944-5956},
	doi={10.1109/TSP.2011.2165064}}

@ARTICLE{DNN,
	author={Labbaf, Navid and Dalili Oskouei, Hamid Reza and Abedi, Mohammad Reza},
	journal={IEEE Trans. Green Commun. Netw.}, 
	title={Robust {DOA} estimation in a uniform circular array antenna with errors and unknown parameters using deep learning}, 
	year={2023},
	volume={7},
	number={4},
	pages={2143-2152},
	doi={10.1109/TGCN.2023.3294448}}

@ARTICLE{HuMIMODeepLearn,
	author={Hu, Die and Zhang, Yonghao and He, Lianghua and Wu, Jun},
	journal={IEEE Wireless Commun. Lett.}, 
	title={Low-complexity deep-learning-based {DOA} estimation for hybrid massive {MIMO} systems with uniform circular arrays}, 
	year={2020},
	volume={9},
	number={1},
	pages={83-86},
	doi={10.1109/LWC.2019.2942595}}

@ARTICLE{YTMIMOGPE,
	author={Tian, Ye and Liu, Wei and Xu, He and Liu, Shuai and Dong, Zhiyan},
	journal={IEEE Trans. Wireless Commun.}, 
	title={{2-D} {DOA} estimation of incoherently distributed sources considering gain-phase perturbations in massive {MIMO} systems}, 
	year={2022},
	volume={21},
	number={2},
	pages={1143-1155},
	doi={10.1109/TWC.2021.3102483}}

@ARTICLE{LiGPEConvex,
	author={Li, Dingkang and Sun, Guohao and Ji, Yuandong},
	journal={IEEE Signal Process. Lett.}, 
	title={{DOA} estimation via distributionally robust optimization with gain-phase errors}, 
	year={2026},
	volume={33},
	number={},
	pages={1101-1105},
	doi={10.1109/LSP.2026.3667071}}

@ARTICLE{WenGPEMIMO,
	author={Wen, Fangqing and Wang, Xianpeng and Gui, Guan and Ohtsuki, Tomoaki and Niyato, Dusit and Adachi, Fumiyuki},
	journal={IEEE Trans. Wireless Commun.}, 
	title={Fast {DOD/DOA} estimation for massive conformal {MIMO} arrays with unknown gain-phase errors}, 
	year={2026},
	volume={25},
	number={},
	pages={15866-15878},
	doi={10.1109/TWC.2026.3686764}}

@ARTICLE{ZhouMUDNA,
	author={Zhou, Lang and Zhou, Gongjian},
	journal={IEEE Internet Things J.}, 
	title={Multi-{ULAs} dilated nested array: a sparse antenna array design with enhanced {DOA} estimation performance}, 
	year={2026},
	volume={13},
	number={14},
	pages={32459-32462},
	doi={10.1109/JIOT.2026.3703787}}

@ARTICLE{LishuangDOF,
	author={Li, Shuang and Zhang, Xiao-Ping},
	journal={IEEE Signal Process. Lett.}, 
	title={A new approach to construct virtual array with increased degrees of freedom for moving sparse arrays}, 
	year={2020},
	volume={27},
	number={},
	pages={805-809},
	doi={10.1109/LSP.2020.2993956}}

@ARTICLE{QinGDDOF,
	author={Qin, Guodong and Zhang, Yimin D. and Amin, Moeness G.},
	journal={IEEE Signal Process. Lett.}, 
	title={{DOA} estimation exploiting moving dilated nested arrays}, 
	year={2019},
	volume={26},
	number={3},
	pages={490-494},
	doi={10.1109/LSP.2019.2894467}}

@ARTICLE{ChenNA,
	author={Chen, Hua and Lin, Hongguang and Liu, Wei and Wang, Qing and Shen, Qing and Wang, Gang},
	journal={IEEE Trans. Signal Process.}, 
	title={Augmented multi-subarray dilated nested array with enhanced degrees of freedom and reduced mutual coupling}, 
	year={2024},
	volume={72},
	number={},
	pages={1387-1399},
	doi={10.1109/TSP.2024.3374557}}

@ARTICLE{LishuangDAfamily,
	author={Li, Shuang and Zhang, Xiao-Ping},
	journal={IEEE Trans. Signal Process.}, 
	title={Dilated arrays: a family of sparse arrays with increased uniform degrees of freedom and reduced mutual coupling on a moving platform}, 
	year={2021},
	volume={69},
	number={},
	pages={3367-3382},
	doi={10.1109/TSP.2021.3083988}}

@ARTICLE{Lishuangq-dilated,
	author={Li, Shuang and Zhang, Xiao-Ping and Ruan, Ding-Liang and Zeng, Hao},
	journal={IEEE Internet Things J.}, 
	title={Increasing the uniform degrees of freedom for moving q-dilated arrays}, 
	year={2023},
	volume={10},
	number={1},
	pages={194-209},
	doi={10.1109/JIOT.2022.3200955}}

@ARTICLE{GuoHDHighOrder,
	author={Guo, Haodong and Chen, Hua and Guo, Muran and Liu, Wei and Tian, Ye and Wang, Gang and So, Hing Cheung},
	journal={IEEE Trans. Signal Process.}, 
	title={Generalized dilated array scheme exploiting high-order virtual co-array for a moving platform}, 
	year={2026},
	volume={74},
	number={},
	pages={878-891},
	doi={10.1109/TSP.2026.3664800}}

@ARTICLE{Convex,
	author={Tirer, Tom and Bialer, Oded},
	journal={IEEE Trans. Aerosp. Electron. Syst.}, 
	title={Direction of arrival estimation and phase-correction for noncoherent subarrays: a convex optimization approach}, 
	year={2022},
	volume={58},
	number={6},
	pages={5571-5585},
	doi={10.1109/TAES.2022.3175465}}

@ARTICLE{YetCalibrateNA,
	author={Tian, Ye and Xu, He},
	journal={IEEE Commun. Lett.}, 
	title={Calibration nested arrays for underdetermined {DOA} estimation using fourth-order cumulant}, 
	year={2020},
	volume={24},
	number={9},
	pages={1949-1952},
	doi={10.1109/LCOMM.2020.3001936}}

@article{YeMAGPE,
	author  = {Chengzhi Ye and Ruoyu Zhang and Changcheng Hu and Lei Yao and Wen Wu},
	title   = {{DOA} estimation for movable antenna array systems under gain and phase errors},
	journal = {Phys. Commun.},
	volume  = {72},
	note = {{Art. no. 102808}},
	year    = {2025},
	doi     = {10.1016/j.phycom.2025.102808}
}

@misc{YeRotable,
	author = {Ye, Chengzhi and Zhang, Ruoyu and Du, Jincheng and
	Ma, Wenyan and Wu, Qingqing and Wu, Wen and Zhang, Rui},
	title = {Rotatable antenna-enhanced wireless sensing with uniform
	sparse array via tensor decomposition},
	howpublished = {\emph{IEEE Wireless Commun. Lett.}, early access},
	month = {Aug. 11,},
	year = {2026},
	note = {doi: 10.1109/LWC.2026.3722674}
}

@misc{YeSelfCalibtate,
	title={Self-calibration {DOA} estimation for movable antenna systems with antenna position errors}, 
	author={Chengzhi Ye and Ruoyu Zhang and Wen Wu and Byonghyo Shim},
	year={2026},
	eprint={2605.23140},
	archivePrefix={arXiv},
	primaryClass={eess.SP},
	url={https://arxiv.org/abs/2605.23140}, 
}

@misc{YeGlobalCom,
	title={Movable antenna enhanced wireless sensing via steering vector correlation and {CRB} optimization}, 
	author={Chengzhi Ye and Ruoyu Zhang and Shichao Wei and Wen Wu and Byonghyo Shim},
	year={2026},
	eprint={2608.09383},
	archivePrefix={arXiv},
	primaryClass={eess.SP},
	url={https://arxiv.org/abs/2608.09383}, 
}

@article{LiaoB1,
	author={Liao, Bin and Chan, Shing Chow},
	journal={IEEE Trans. Antennas Propag.}, 
	title={Direction finding with partly calibrated uniform linear arrays}, 
	year={2012},
	volume={60},
	number={2},
	pages={922-929},
	doi={10.1109/TAP.2011.2173144}}

@ARTICLE{YeRobust,
	author={Ye, Chengzhi and Zhang, Ruoyu and Hu, Changcheng and Yao, Lei and Wu, Wen and Yuen, Chau},
	journal={IEEE Wireless Commun. Lett.}, 
	title={Robust {DOA} estimation for movable antenna arrays with partial gain and phase errors}, 
	year={2026},
	volume={15},
	number={},
	pages={545-549},
	doi={10.1109/LWC.2025.3630221}}

@book{Kay1993,
	title={{F}undamentals of {S}tatistical {S}ignal {P}rocessing: {E}stimation {T}heory},
	author={Kay, Steven M.},
	year={1993},
	publisher={Prentice Hall},
	address={Upper Saddle River, NJ, USA}
}

@ARTICLE{10522654,
	author={Molaei, Amir Masoud and Zakeri, Bijan and Andargoli, Seyed Mehdi Hosseini and Abbasi, Muhammad Ali Babar and Fusco, Vincent and Yurduseven, Okan},
	journal={IEEE Access}, 
	title={A comprehensive review of direction-of-arrival estimation and localization approaches in mixed-field sources scenario}, 
	year={2024},
	volume={12},
	number={},
	pages={65883-65918},
	doi={10.1109/ACCESS.2024.3398351}}

@ARTICLE{526899,
	author={Krim, H. and Viberg, M.},
	journal={IEEE Signal Process. Mag.}, 
	title={Two decades of array signal processing research: the parametric approach}, 
	year={1996},
	volume={13},
	number={4},
	pages={67-94},
	doi={10.1109/79.526899}}

@ARTICLE{10950424,
	author={Mansourian, Amir and Fadakar, Alireza and Akhavan, Saeed and Maham, Behrouz},
	journal={IEEE Open J. Commun. Soc.}, 
	title={Robust {3-D} multi-source localization with a movable antenna array via sparse signal processing}, 
	year={2025},
	volume={6},
	number={},
	pages={3664-3682},
	doi={10.1109/OJCOMS.2025.3558476}}

@ARTICLE{10529957,
	author={Chen, Hui and Keskin, Musa Furkan and Sakhnini, Adham and Decarli, Nicolò and Pollin, Sofie and Dardari, Davide and Wymeersch, Henk},
	journal={IEEE Wireless Commun.}, 
	title={{6G} localization and sensing in the near field: features, opportunities, and challenges}, 
	year={2024},
	volume={31},
	number={4},
	pages={260-267},
	doi={10.1109/MWC.011.2300359}}

@ARTICLE{10144718,
	author={Pesavento, Marius and Trinh-Hoang, Minh and Viberg, Mats},
	journal={IEEE Signal Process. Mag.}, 
	title={Three more decades in array signal processing research: an optimization and structure exploitation perspective}, 
	year={2023},
	volume={40},
	number={4},
	pages={92-106},
	doi={10.1109/MSP.2023.3255558}}

@ARTICLE{665,
	author={Van Veen, B. D.  and Buckley, K. M. },
	journal={IEEE ASSP Mag.}, 
	title={Beamforming: a versatile approach to spatial filtering}, 
	year={1988},
	volume={5},
	number={2},
	pages={4-24},
	doi={10.1109/53.665}}

@ARTICLE{8809740,
	author={Wang, Peng and Kong, Yujun and He, Xuefang and Zhang, Mingxing and Tan, Xiuhui},
	journal={IEEE Access}, 
	title={An improved squirrel search algorithm for maximum likelihood {DOA} estimation and application for {MEMS} vector hydrophone array}, 
	year={2019},
	volume={7},
	number={},
	pages={118343-118358},
	doi={10.1109/ACCESS.2019.2936823}}

@INPROCEEDINGS{8212844,
	author={Vikas, B. and Vakula, D.},
	booktitle={Proc. Int. Conf. Electron., Commun. Aerosp. Technol. (ICECA)}, 
	title={Performance comparision of {MUSIC} and {ESPRIT} algorithms in presence of coherent signals for {DOA} estimation}, 
	year={2017},
	volume={2},
	number={},
	pages={403-405},
	doi={10.1109/ICECA.2017.8212844}}

@ARTICLE{5456168,
	author={Pal, Piya and Vaidyanathan, P. P.},
	journal={IEEE Trans. Signal Process.}, 
	title={Nested arrays: a novel approach to array processing with enhanced degrees of freedom}, 
	year={2010},
	volume={58},
	number={8},
	pages={4167-4181},
	doi={10.1109/TSP.2010.2049264}}

@ARTICLE{9773022,
	author={Shaalan, Ahmed M. A. and Du, Jun and Tu, Yan-Hui},
	journal={IEEE Trans. Signal Process.}, 
	title={Dilated nested arrays with more degrees of freedom ({DOFs}) and less mutual coupling—{Part} {I}: the fundamental geometry}, 
	year={2022},
	volume={70},
	number={},
	pages={2518-2531},
	doi={10.1109/TSP.2022.3174451}}
\end{document}